\documentclass[nohyperref]{article}

\usepackage{soul}
\usepackage{microtype}
\usepackage{graphicx}
\usepackage{subfigure}
\usepackage{booktabs} 
\usepackage{subfigure}
\usepackage{hyperref}
\usepackage{comment}

\usepackage[accepted]{icml2023}

\usepackage{amsmath}
\usepackage{amssymb}
\usepackage{mathtools}
\usepackage{amsthm}
\usepackage{physics}

\usepackage[capitalize,noabbrev]{cleveref}
\usepackage{mathtools}

\DeclarePairedDelimiterX\Set[2]{\{}{\}}{#1\,\delimsize\vert\,#2}

\theoremstyle{plain}
\newtheorem{theorem}{Theorem}[section]

\newtheorem{lemma}[theorem]{Lemma}
\newtheorem{corollary}[theorem]{Corollary}
\theoremstyle{definition}

\theoremstyle{remark}

\def\bE{\mathbf{E}}

\def\bI{\mathbf{I}}

\def\bN{\mathbf{N}}

\def\bP{\mathbf{P}}

\def\bR{\mathbf{R}}

\newcommand{\calO}{{\mathcal{O}}}

\newcommand{\calS}{{\mathcal{S}}}

\newcommand{\df}{\mathrm{d}}
\newcommand{\Geo}{{\mathsf{Geo}}}
\newcommand{\Se}{S^{\mathsf{E}}}
\newcommand{\So}{S^{\mathsf{O}}}
\newcommand{\xio}{\xi^{\mathsf{O}}}
\newcommand{\xie}{\xi^{\mathsf{E}}}
\newcommand{\READ}{\mathsf{READ}}
\renewcommand{\var}{\text{Var}}

\def\y#1{y^{(#1)}}
\def\yto#1{y^{(0:#1)}}
\def\yfromto#1#2{y^{(#1:#2)}}

\usepackage[textsize=tiny]{todonotes}

\icmltitlerunning{Optimal randomized multilevel Monte Carlo for repeatedly nested expectations}
\begin{document}

\twocolumn[
\icmltitle{Optimal randomized multilevel Monte Carlo for repeatedly nested expectations}




\begin{icmlauthorlist}
\icmlauthor{Yasa Syed}{Rutgers}
\icmlauthor{Guanyang Wang}{Rutgers}

\end{icmlauthorlist}

\icmlaffiliation{Rutgers}{Department of Statistics, Rutgers University, New Brunswick, United States}

\icmlcorrespondingauthor{Guanyang Wang}{guanyang.wang@rutgers.edu}

\icmlkeywords{nested expectation, optimal estimator, Multilevel Monte Carlo}

\vskip 0.3in
]


\printAffiliationsAndNotice{}  

\begin{abstract}
The estimation of repeatedly nested expectations is a challenging task that arises in many real-world systems. However, existing methods generally suffer from high computational costs when the number of nestings becomes large. Fix any non-negative integer $D$ for the total number of nestings. Standard Monte Carlo methods typically cost at least  $\calO(\varepsilon^{-(2+D)})$ and sometimes $\calO(\varepsilon^{-2(1+D)})$ to obtain an estimator up to  $\varepsilon$-error.  More advanced methods, such as multilevel Monte Carlo,  currently only exist for $D = 1$. In this paper, we propose a novel Monte Carlo estimator called $\READ$, which stands for ``Recursive Estimator for Arbitrary Depth.'' Our estimator has an optimal computational cost of $\calO(\varepsilon^{-2})$  for every fixed $D$ under suitable assumptions, and a nearly optimal computational cost of $\calO(\varepsilon^{-2(1 + \delta)})$  for any $0 < \delta < \frac12$ under much more general assumptions.  Our estimator is also unbiased, which makes it easy to parallelize. The key ingredients in our construction are an observation of the problem's recursive structure and the recursive use of the randomized multilevel Monte Carlo method.
\end{abstract}

\section{Introduction}
\label{sec:intro}

\medskip
Monte Carlo methods are a class of algorithms that use random sampling to estimate quantities of interest, such as integrals or expected values. When the estimand can be expressed as  an expectation, for example $\bE_\pi[g(X)]$, these methods work by generating independent random samples $X_1, \ldots, X_n$ from $\pi$, and using the average $\sum_{i=1}^ng(X_i)/n$ as an estimator. Monte Carlo estimators are unbiased and converge at a rate of $n^{-1/2}$, regardless of the dimension of the samples. This dimension-independent convergence rate makes Monte Carlo methods a powerful tool for approximating high-dimensional integrations, as they do not suffer from the curse of dimensionality that plagues deterministic numeric integration methods.

However, the above analysis implicitly assumes the integrand $g$ can be pointwisely evaluated, which may not be possible in many situations. This can arise, for instance, when it is expressed as another integration over latent variables  or when it involves solving a optimization problem. In this paper, we study the problem of estimating repeatedly nested expectations (RNE), which means the integrand  depends on a sequence of other functions and conditional expectations. Specifically, fix any positive integer $D$ for the total number of nestings, and $\{g_d\}_{d=0}^D$ for a family of real-valued functions which can be pointwisely evaluated. Let $(\y0, \ldots, \y{D})$ be a finite-time stochastic process with underlying joint distribution $\pi$, and let $\y{0:d}$ denote the vector $(\y{0}, \ldots, \y{d})$ for every $d\leq D$. The RNE, first formally formulated in \cite{rainforth2018nesting}, is defined as: 
\begin{align} \label{eqn:nested target}
    \gamma_0
=
\bE
\left[ 
g_0
\left(
\y{0},
\gamma_1
\left(
\y0
\right)
\right)
\right],
\end{align}

where $\{\gamma_i\}_{i=1}^{D-1}$ is recursively defined as:
\begin{align}\label{eqn: nesting}
    \gamma_d(\yto{d-1})
=
\bE
\left[
g_d
\left(
\yto{d}
,
\gamma_{d+1}
\left(
\yto{d}
\right)
\right)
\bigg|~
\yto{d-1}
\right],
\end{align}

and 
\begin{align}\label{eqn: nested final level}
\gamma_D(\yto{D-1})
=
\bE
\left[
g_D
\left(
\yto{D}
\right) 
\bigg|~ \yto{D-1}
\right].
\end{align}

The estimation of Resource-Optimal Nested Expectations (RNEs) is a significant challenge that encompasses various practical scenarios, where the desired outcome relies on multiple stages or decision points. Here, we provide several instances to illustrate this:
\begin{itemize}
    \item In financial modeling, one crucial problem involves estimating RNEs when $\gamma_0$ represents the expected utility of an optimal strategy in a $D$-horizon optimal stopping problem. Here, $g_d(\yto{d}, u)$ is defined as $\max{\y{d}, u}$ for $0\leq d \leq D-1$, and $g_D(\yto{D})$ is simply $\y{D}$.

\item When $D = 2$, a recent paper by \cite{giles2023efficient} focuses on risk estimation for the credit valuation adjustment. In their analysis, the outermost function $g_0$ is a Heaviside function, while the inner functions $g_1$ and $g_2$ are smooth functions.

\item When $D = 1$, RNE estimation finds extensive applications in Bayesian experimental design \cite{goda2022unbiased}, portfolio risk management \cite{gordy2010nested}, stochastic and bilevel optimization \cite{hu2021bias}, as well as variational Bayes \cite{he2022unbiased}.
\end{itemize}

In addition to the aforementioned examples, RNE estimation, sometimes also referred to nonlinear Monte Carlo, finds relevance in various fields including probabilistic programs \cite{rainforth18nesting_prob}, numerical partial differential equations (PDEs) \cite{beck2020nonlinear}, as well as physics and chemistry \cite{dauchet2018addressing}.

However, estimating RNEs is challenging. As shown in formulas \eqref{eqn:nested target} -- \eqref{eqn: nested final level}, we are interested in the expectation of $g_0$, which depends  on the random variable $\y0$ and $\gamma_1(\y0)$ -- a conditional expectation of $g_1$ given $\y0$. Then $g_1$  further depends on a random variable $\y1$ and $\gamma_2(\y0, \y1)$ which is a  conditional expectation of $g_2$ given $\y0$ and $\y1$.  This procedure is recursively defined until it reaches the deepest depth, $D$. Since $\gamma_1(\y0)$ (and also $\gamma_2, \gamma_3, \ldots$) cannot be directly evaluated in most practical cases, estimating RNEs cannot be handled by  standard Monte Carlo methods. 

The most natural way to estimate RNEs is by nesting  Monte Carlo (NMC) estimators. In the $D = 1$ case, this method works by first sampling independent and identically distributed (i.i.d.) copies $\y{0}_1, \ldots, \y{0}_{N_0}$ according to the distribution of $\y0$. For each fixed $\y0_i$, one further samples $N_1$ i.i.d. $\y1_1, \ldots, \y1_{N_1}$ according to $\pi(\y1\mid \y0_i)$, and uses the standard estimator
$ \hat{\gamma}_1(\y0_i) \coloneqq \sum_{i=j}^{N_1} g_1(\y0_i,\y1_j)/N_1$ to estimate $\gamma_1(\y0_i)$. The final estimator uses the estimated $\hat{\gamma}_1(\y0_i)$ to replace the intractable $\gamma_1(\y0_i)$, i.e.,
$$I_{N_0, N_1}= \frac{1}{N_0}\sum_{i=1}^{N_0} g_0(\y0_i, \hat{\gamma_1}(\y0_i)).$$ This nested estimator can be easily extended to the general $D$ case, albeit the notations become more complex. Roughly, one still samples $N_0$ i.i.d. copies according to $\pi(\y0)$, and for each fixed trajectory $\yto{d-1}$,  the user generates  $N_d$   i.i.d. samples from $\pi(\y{d}\mid \yto{d-1})$ all the way to depth $D$ and then form the nested estimator from the deepest depth to the shallower depths. The construction details are referred to Section 3.2 of \cite{rainforth2018nesting}. 

After suitably allocating the number of samples $(N_d)_{d=0}^D$ for each depth, the root-mean-square error (rMSE) of the NMC estimator converges to $0$ at a rate of $N^{-1/(2D+2)}$ or $N^{-1/(D+2)}$ \cite{rainforth2018nesting}, depending on the regularity conditions of the functions $\{g_d\}_{d=0}^D$, where $N = \prod_{d=0}^D N_d$ is the total number of samples used to form a nested estimator. This convergence rate diminishes exponentially with $D$, meaning that NMC estimators do not have the same dimension-free convergence rate as standard Monte Carlo estimators.   As a result, NMC methods require at least $\calO(\varepsilon^{-(2+D)})$ and sometimes $\calO(\varepsilon^{-2(1+D)})$ samples to get an estimator within $\varepsilon$-rMSE, while standard Monte Carlo estimators require only $O(\varepsilon^{-2})$ samples. Although there are a few cases mentioned in \cite{rainforth2018nesting} where the canonical $\calO(N^{-1/2})$ rate can be achieved, the problem of estimating RNEs with an optimal (or dimension-free) convergence rate remains largely open.

In the special case $D = 1$, more efficient methods have been proposed \cite{giles2018mlmc,giles2019decision,giles2019multilevel} based on the celebrated multilevel Monte Carlo (MLMC) methods \cite{heinrich2001multilevel, giles2008multilevel}. These estimators achieve up to $\varepsilon$-rMSE with cost $O(\varepsilon^{-2} \log(1/\varepsilon)^2)$ or $O(\varepsilon^{-2})$ under varying conditions, comparing favorably with the NMC estimator. However, existing methods cannot be directly generalized to solve the general $D$ case. Meanwhile, implementing these methods requires users to prespecify the precision level $\varepsilon$ and conduct preliminary experiments/calculations to carefully estimate/bound the parameters in the MLMC algorithm (see, e.g., Theorem 1 of \cite{giles2019decision}). Therefore, existing MLMC estimators seem to be harder to implement and less amendable to our original problem, which has a recursive structure.

In this work, we propose the $\READ$, a novel Monte Carlo estimator for the RNE estimation with an arbitrary number of nestings $D$. Our construction is interesting in the following three aspects. Firstly, under suitable regularity conditions similar to those in \cite{rainforth2018nesting}, the rMSE of our estimator has an optimal convergence rate $N^{-1/2}$ regardless of $D$. Equivalently, our method costs in expectation $\calO(\varepsilon^{-2})$ to get an estimator up to $\varepsilon$-rMSE. Under much more general assumptions, our method still achieves a nearly-optimal cost  of $\calO(\varepsilon^{-2(1 + \delta)})$  for any $0 < \delta < \frac12$ to get an estimator up to $\varepsilon$-mean-absolute-error (MAE).

It is worth mentioning that most of our effort is devoted to designing unbiased estimators of $\gamma_0$ in \eqref{eqn:nested target} with finite computational cost and finite variance (or finite (2-$\delta$)-th moment under more general assumptions). After developing such an unbiased estimator, we can simulate independent copies of these estimators and average them. The $N^{-1/2}$ convergence rate and $\mathcal{O}(\varepsilon^{-2})$  cost are then immediate corollaries of the bias-variance decomposition formula, see Corollary \ref{cor: computational cost, LBS case}.

Therefore, another appealing property of $\READ$, in contrast to existing methods, is that it admits no estimation bias. Unbiasedness implies these estimators can be implemented in parallel processors without requiring any communication between them. Designing unbiased estimators has recently attracted much interest in statistics, operations research, and machine learning communities for its potential for parallelization. Our methods add to the rich body of works of \cite{glynn2014exact,rhee2015unbiased,Blanchet2015UnbiasedMC,jacob2020unbiased,biswas2019estimating,wang2021maximal, wang2022unbiased,kahale2022unbiased}.  

Finally, our algorithm for constructing $\READ$ relies on the randomized multilevel Monte Carlo (rMLMC) method \cite{mcleish2011general,rhee2015unbiased,Blanchet2015UnbiasedMC}, but it is significantly different from its previous applications. Many of the current applications of randomized Multilevel Monte Carlo (rMLMC) methods \cite{rhee2015unbiased,vihola2018unbiased,goda2022unbiased} also have a deterministic version known as the original Multilevel Monte Carlo (MLMC) \cite{giles2008multilevel}, which offers similar or even better guarantees in terms of computational cost. As a result, it is natural to speculated that every problem solved by rMLMC has a corresponding deterministic version. However, our findings indicate that this assumption may not always be accurate. The rMLMC framework is well-suited to the recursive structure of RNEs, and can be used as a subroutine in our method. In contrast, the non-randomized MLMC cannot be easily applied to the general case of $D > 1$. This suggests that the rMLMC framework may be more widely applicable than previously thought.

The rest of this paper is organized as follows: in the remainder of this section, we discuss related works, set up our notation, and introduce our technical assumptions. In Section \ref{sec:our approach}, we introduce our algorithm and show that it attains the optimal and nearly-optimal computational cost under two different assumptions, respectively. In Section \ref{sec:numerical}, we demonstrate the empirical performance of our method on several toy examples. We conclude this paper with a short discussion in Section \ref{sec:discussion}. Proof and experiment details are deferred to the Appendix. An additional experiment is  also included in Appendix \ref{sec:extra experiment}.

\subsection{Related work}\label{subsec:related_work}
Our algorithm design strategy mainly follows the randomized multilevel Monte Carlo (rMLMC) framework \cite{mcleish2011general,rhee2015unbiased,Blanchet2015UnbiasedMC}.  Our algorithm is  inspired by the unbiased optimal stopping estimator \cite{zhou2022unbiased}, which develops  estimators for the optimal stopping problem by recursively calling the rMLMC algorithm. We extend the methodology in \cite{zhou2022unbiased} both in scope and depth. Our method works with a more general class of problems formulated by \cite{rainforth2018nesting}, which includes the optimal stopping problem as a special case, and provides more precise results under practical assumptions.

Throughout this paper, we will assume the functions $\{g_d\}_{d=0}^{D-1}$ are all continuous and the  process $\pi$ can be perfectly simulated.  When $D = 1$ and $g_0$ is discontinuous, progress has been made by \cite{broadie2011efficient} and \cite{giles2019multilevel,giles2022multilevel}. When the underlying distribution is itself challenging, users have to first use MCMC to approximately sample from $\pi$. The case of $D = 1$ and challenging $\pi$ is considered in \cite{wang2022unbiased}.   

\subsection{Notations}\label{subsec:notations}
Now we introduce our notations. Many of our notations follow those used in the original definition   \cite{rainforth2018nesting}, despite generalizing their setting to a multivariate underlying process. Throughout this paper, we preserve the letter $D$ for the total number of nestings. We denote by $\pi$ the underlying joint distribution of a finite-time, real-valued, $M$-dimensional stochastic process $(\y0, \ldots, \y{D})$, i.e. $\y{d} \in \bR^M$ for each $0 \leq d \leq D$.

For every $0\leq i \leq j \leq D$, we use the \ $\yfromto{i}{j}$ to denote the vector $(\y{i}, \ldots, \y{j})$. The conditional distribution of $\yfromto{d}{D}$ given the value of $\yto{d-1}$ is denoted by $\pi_{d:D}(\cdot \mid \yto{d-1})$. The marginal distribution of $\y{d}$ conditioning on $\yto{d-1}$ is denoted by $\pi_d(\cdot\mid \yto{d-1})$. We adopt the convention that $\yto{-1} = \varnothing$, and therefore $\pi_0$ stands for the (unconditioned) marginal distribution of $\y{0}$. Let $\Pi$ be any probability distribution on some probability space, and $Z$ be some random variable on the same space, then we use $\lVert Z\rVert_{\Pi,m}$ to denote the $L^m$--norm of $Z$ under $\Pi$, i.e., $\big(\bE_{\Pi}[\lvert Z\rvert^m]\big)^{1/m}$. The geometric distribution with parameter $r$ is denoted by $\Geo(r)$. We also define $p_r(n) \coloneqq \bP[\Geo(r) = n] = r (1-r)^n$ for every $n \in \{0,1,2,\ldots, \}$. For every $0\leq d \leq D-1$, the function $g_d$ introduced in  \eqref{eqn:nested target} -- \eqref{eqn: nesting} maps from $\bR^{(d+1)M+1}$ to $\bR$ since $g_d$ takes as its first $d+1$ arguments $M$-dimensional vectors, and it takes only a scalar as its final argument. The function $g_D$ in \eqref{eqn: nested final level} maps from $\bR^{(D+1)M}$ to $\bR$ since $g_D$ has all $D+1$ vectors in the $M$-dimensional process as its arguments. For random variables $X_1, \ldots, X_{n}$, we denote their summation by $S_n \coloneqq \sum_{i=1}^n X_i$. When $n$ is even, we denote by $\So_{n/2} \coloneqq \sum_{k=1}^{n/2} S_{2k-1}$ and $\Se_{n/2} \coloneqq \sum_{k=1}^{n/2}S_{2k}$ the summations of their odd and even terms, respectively. 

\subsection{Assumptions}\label{subsec:assumptions}
Throughout this paper, we assume that we can access a simulator $\cal S$. The simulator can take any trajectory $\yto{d-1}$ with $0 \leq d\leq D$ as input, and outputs $\y{d}$ which follows the distribution $\pi_d(\cdot\mid \yto{d-1})$. In particular, $\calS$ can take $\varnothing$ as input and simulates $\y{0}\sim \pi_0$. Calling $\calS$ recursively for $D+1$ times generates one complete sample path. This assumption enables us to sample from any marginal or conditional distribution perfectly. This assumption is also standard and is posed explicitly or implicitly in nearly all the existing works concerning the estimation of nested expectations, see \cite{giles2019decision, goda2022unbiased, zhou2022unbiased} for examples.

For $0 \leq d \leq D - 1$, fix $g_d : \bR^{(d+1)M+1} \to \bR$. We say $g_d$ satisfies the last-component bounded second derivative condition (LBS) if there exists a $K_d < \infty$ such that
\begin{align}\label{eqn: last-second-derivative}
\sup_{(\yto{d}, z)} 
\left|\partial^2_{(d+1)M + 1} g_d(\yto{d}, z)\right| < K_d.
\end{align}
We say $g_d$ satisfies the last-component bounded Lipschitz condition (LBL) if there exists an $L_d < \infty$ such that for all $x, z \in \bR$
\begin{align}\label{eqn: last-lipschitz}
\sup_{\yto{d}}
|g_d(\yto{d}, x) - g_d(\yto{d}, z) < L_d|x - z|.
\end{align}
These assumptions (and their variants) are also posed in related works such as \cite{rainforth18nesting_prob, Blanchet2015UnbiasedMC, giles2018mlmc}.

\section{Algorithm, estimator, and theoretical results}\label{sec:our approach}
Now we are ready to present our main results. As discussed in Section \ref{sec:intro}, we will be focusing on designing a Monte Carlo estimator which is unbiased, has a finite computational cost, and has finite variance or (2-$\delta$)-th moment under different assumptions. 

\subsection{Preliminary analysis}\label{subsec:preliminary}

One of the challenges in estimating the RNEs is the difficulty of estimating $\gamma_1(\y0)$. Users typically first estimate $\gamma_1(\y0)$ and then use these estimators to estimate $\gamma_0$. For the time being, we are temporarily adding the assumption that users can simulate unbiased estimators $\hat\gamma_1(\y0)$ of $\gamma_1(\y0)$ for every fixed $\y0$ with finite computational cost. This assumption will be removed in Section \ref{subsec:general case}. It easily holds when $D = 1$, as users can repeatedly simulate $\y1_i \sim \pi_1(\cdot \mid \y0)$ and it follows from the problem definition that each $g_1(\y0,\y1_i)$ is unbiased for $\gamma_1(\y0)$. In the general case of $D > 1$, this assumption is far from trivial, as $\gamma_1(\y0)$ is itself a nested expectation with a nesting depth of $D-1$. Nevertheless, as we will see in Section \ref{subsec:general case}, this assumption helps us to capture and reduce the intrinsic difficulty of the problem and, therefore, will guide us to design the general algorithm.    

With this extra assumption, constructing unbiased estimators of  \eqref{eqn:nested target} is  equivalent to constructing unbiased estimators of $g_0(\y0, \gamma_1(\y0))$. Even with access to unbiased estimators of $\gamma_1(\y0)$, the intuitive plug-in estimator $g_0\big(\y0, \hat\gamma_1(\y0)\big)$ is still biased, as in general $\bE[g_0\big(\y0, \hat\gamma_1(\y0)\big) \mid \y0] \neq g_0(\y0, \bE[\hat\gamma_1(\y0)\mid \y0])$. To eliminate this bias, we use the rMLMC method \cite{Blanchet2015UnbiasedMC}, which is briefly reviewed below.

The rMLMC method uses the Law of Large Numbers (LLN) and rewrites $g_0$ as the following telescoping summation. 
\begin{align*}
   &g_0(\y0, \gamma_1(\y0)) 
   = 
   \bE
   \left[
   g_0\left(\y0, \lim_{k\rightarrow \infty}\frac{S_k}{k}\right) \bigg|~ \y0
   \right]\\
    &= \sum_{n=1}^{\infty} 
    \bE
    \left[
    g_0\left(\y0,  \frac{S_{2^n}}{2^n}\right) \bigg|~\y0
    \right] \\
    &\quad\quad\quad\quad- 
    \bE
    \left[
    g_0\left(\y0,  \frac{S_{2^{n-1}}}{2^{n-1}}\right) \bigg|~\y0
    \right],
\end{align*}
where $S_k = \sum_{i=1}^k \hat\gamma_{1,i}(\y0)$ is the summation of i.i.d. copies of $\hat\gamma_{1}(\y0)$. To unbiasedly estimate the infinite sum, the rMLMC algorithm  first samples $\y0\sim \pi_0$, then  samples a random $N\sim \Geo(r)$, finally generates $2^N$ unbiased estimators  $\{\hat\gamma_{1,i}(\y0)\}_{i=1}^{2^N}$ of $\gamma_1(\y0)$ and estimates $\gamma_0$ by  $R_0\coloneqq \Delta_N/p_r(N)$, where $\Delta_n$ is defined as: 
\begin{align*}
    \Delta_n \coloneqq g_0\left(\y0,  \frac{S_{2^n}}{2^n}\right) - \frac{1}{2}\Bigg[&g_0\bigg(\y0,  \frac{\Se_{2^{n-1}}}{2^{n-1}}\bigg) \\ & + g_0\bigg(\y0,  \frac{\So_{2^{n-1}}}{2^{n-1}}\bigg)\Bigg]
\end{align*}
for $n\geq 1$ and $\Delta_0 \coloneqq g_0(\y0, \hat \gamma_{1,1}(\y0))$.

The next theorem justifies the theoretical properties of $R_0$:
\begin{theorem}\label{thm:nesting_D=1}
With all the notations as above, suppose $g_0: \bR^{M+1}\rightarrow \bR$ satisfies LBS condition defined in \eqref{eqn: last-second-derivative}, and $\lVert \hat\gamma_1(\y0) \rVert_{\pi, m} < \infty$ for some $m \geq 4$. Then for any $r \in (1/2, 3/4)$, the estimator  $R_0 \coloneqq \Delta_N/p_r(N)$ has expectation $\gamma_0$, finite variance, and finite expected computational cost. 
\end{theorem}

Theorem \ref{thm:nesting_D=1} will  be proved  as  a special case of our  Theorem \ref{thm:general D, second order case}. For now, we use the following heuristic calculation to justify the unbiasedness of $\hat\gamma_0$: 
\begin{align*}
    &\bE
    [ R_0\big|~\y0 ] 
     = \sum_{n=0}^\infty\bE\left[\frac{\Delta_n}{p_r(n)} p_r(n)\big|~ \y0\right]\\  
    &= \sum_{n=0}^{\infty} \bE\left[g_0\left(\y0,  \frac{S_{2^n}}{2^n}\right)  - g_0\left(\y0,  \frac{S_{2^{n-1}}}{2^{n-1}}\right) \bigg|~\y0 \right]\\
    & = g_0(\y0, \gamma_1(\y0)).
\end{align*}
Therefore $\bE[R_0] = \bE[g_0(\y0, \gamma_1(\y0))] = \gamma_0$ by \eqref{eqn:nested target}. More technical discussions such as the range of $r$, other possible regularity conditions on $g_0$, and the moment guarantees of $\gamma_0$ will all be deferred after Theorem \ref{thm:general D, second order case}.

\subsection{Recursive rMLMC algorithm for general $D$}\label{subsec:general case}

Theorem \ref{thm:nesting_D=1} is useful to solve our original problem (without extra assumptions) in two ways. First, Theorem \ref{thm:nesting_D=1} already solves the case where $D=1$, as our extra assumption automatically holds. It states that if $g_0$ has a bounded second derivative on its last component, and $g_1(\y0,\y1)$ has at least finite fourth moment under $\pi$, then $R_0$ is unbiased, has finite variance, and finite expected computational cost. More importantly, Theorem \ref{thm:nesting_D=1} tells us that the original problem of estimating an RNE with a depth of $D$ can be solved if we can unbiasedly estimate $\gamma_1(\y0)$ for  fixed $\y0$, which is another RNE with a depth of $D-1$. Therefore, we have successfully reduced the number of nestings by one. This observation motivates us to come up with an algorithm for the general $D$ case, as explained below.

We first go one step further to illustrate the $D =2$ case.  When $D = 2$, estimating $\gamma_1(\y0)$  reduces to the case we have analyzed in Section \ref{subsec:preliminary}. To be precise, since $g_2(\yto 2)$ is unbiased for $\gamma_2(\yto 1)$ if $\y2\sim \pi_2(\cdot \mid \yto1)$, one can first sample $\y1\sim \pi_1(\cdot \mid \y0)$, then simulate $N\sim \Geo(r)$ and $2^N$ samples $\{\y2_i\}_{i=1}^{2^N}$ from $\pi_2(\cdot \mid \yto1)$. Let $\hat\gamma_{2,i}(\yto1) \coloneqq g_2(\yto1, \y2_i)$, our estimator of $\gamma_1(\y0)$ is then constructed in the same way as Section \ref{subsec:preliminary}, i.e., $R_1(\y0) := \Delta_N/p_r(N)$ with 
\begin{align*}
    \Delta_n \coloneqq g_1\left(\yto 1,  \frac{S_{2^n}}{2^n}\right) - \frac{1}{2}\Bigg[&g_1\bigg(\yto1,  \frac{\Se_{2^{n-1}}}{2^{n-1}}\bigg) \\ & + g_1\bigg(\yto1,  \frac{\So_{2^{n-1}}}{2^{n-1}}\bigg)\Bigg],
\end{align*}
where $S_{2^n}, \Se_{2^{n-1}},\So_{2^{n-1}}$ are the summation of every, even, and odd terms of $\{\hat{\gamma}_{2,i}(\yto1)\}$, respectively. The same procedure of simulating $R_1(\y0)$ can be repeated independently. Therefore we can  sample  another geometrically distributed random variable $N'\sim \Geo(r')$, and generate $R_{1,i}(\y0) := \Delta_{N'}/p_r(N')$ independently. Since each $R_{1,i}(\y0)$ is unbiased for $\gamma_1(\y0)$, one can again use the method described in Section \ref{subsec:preliminary} to form our final estimator for $\gamma_0$. After checking  $R_{1}(\y0)$ satisfies the finite fourth-moment assumption, Theorem \ref{thm:nesting_D=1} can be applied which implies our estimator is unbiased, has finite variance and finite cost (for the $D = 2$ case). 

The general case works in the same way. A key observation is that, due to the nested structure of the problem, Theorem \ref{thm:nesting_D=1} not only states that an unbiased estimator of $\gamma_0$ can be constructed if one can unbiasedly estimate $\gamma_1(\y0)$ for every $\y0$, but also directly implies that an unbiased estimator of $\gamma_{d}(\yto{d-1})$ can be constructed if one can unbiasedly estimate $\gamma_{d+1}(\yto{d})$ for every $\yto{d}$. Therefore, we can estimate $\gamma_0$ in a backward, inductive manner.

To begin, we consider the deepest depth of the problem, fixing any $\yto{D-1}$. An unbiased estimator of $\gamma_D(\yto{D-1})$ can be directly constructed as $g_D(\yto{D-1}, \y{D})$, where $\y{D} \sim \pi_{D}(\cdot \mid \yto{D-1})$. For $0 \leq d \leq D-1$, if we assume that users can generate unbiased estimators of $\gamma_{d+1}(\yto{d})$ for every $\yto{d}$, then we can obtain an unbiased estimator of $\gamma_{d}(\yto{d-1})$ by sampling one $\y{d}$, generating $N_{d} \sim \Geo(r_{d})$ and $2^{N_{d}}$ unbiased estimators of $\gamma_{d+1}(\yto{d})$, and applying the method described in Section \ref{subsec:preliminary}. This process continues  until we reach $d = 0$, at which point we have an unbiased estimator of $\gamma_0$. The parameters $(r_0, r_1, \ldots, r_{D-1})$ will be carefully chosen and depend on the regularity assumptions of $(g_0, g_1, \ldots, g_{D-1})$. These choices will be discussed in more detail later.

Our algorithm is described in Algorithm \ref{alg:recursive-rMLMC}. It is written as a recursive algorithm, though it could also be equivalently written in an iterative form with much more cumbersome notations. Algorithm \ref{alg:recursive-rMLMC} takes a depth index, a trajectory, a simulator, and parameters for the geometric distribution as inputs, and outputs an unbiased estimator of $\gamma_d(H)$. In particular, with  inputs \{depth $ = 0$, trajectory = $\varnothing$, parameters = $(r_0, r_1, \ldots, r_{D-1})$\}, it outputs $\READ$ -- an unbiased estimator of the RNE defined in \eqref{eqn:nested target}. The logic of Algorithm \ref{alg:recursive-rMLMC} is precisely the same as we just discussed. To estimate $\gamma_d(\yto{d-1})$, the algorithm first checks the value of $d$. When $d = D$, the problem becomes straightforward. When $d < D$, the algorithm samples $\y{d}$, appends $\y{d}$ to the trajectory, samples $N_d$, and calls itself $2^{N_d}$ times with depth $d + 1$ and new trajectory $\{\yto{d}\}$ to get $2^{N_d}$ unbiased estimators of $\gamma_{d+1}(\yto{d})$. Finally, we split these $2^{N_d}$ estimators into even and odd terms and apply the method described  in Section \ref{subsec:preliminary}. The algorithm is guaranteed to stop, as the depth will eventually reach the deepest depth $D$. 

\begin{algorithm}[tb]
   \caption{ A recursive rMLMC algorithm for  RNEs}
   \label{alg:recursive-rMLMC}
\begin{algorithmic}
   \STATE {\bfseries Input:} Depth index $d \in \{0,..., D\}$. Trajectory history $H = \{y^0,... ,y^{d-1}\}$ or $\varnothing$. A simulator $\mathcal{S}$. Parameters $r_{d}, ..., r_{D - 1}$ determined by conditions on $\{g_d\}_{d=0}^{D-1}$ (see Theorem \ref{thm:general D, second order case}, \ref{thm:general D, Lipschitz case}).
   \STATE {\bfseries Output}: An unbiased estimator of $\gamma_d(H)$
    \IF{$d = D$}
      \STATE Sample one $y^{(D)} \sim \pi_D\left( \cdot \mid H \right)$;
      \STATE {\bfseries Return:} $R_D \coloneqq g_D\left(\yto{D} \right).$ 
    \ELSE
    \STATE Sample $y^{(d)} \sim \pi_d\left( \cdot \mid H \right)$;
\STATE Update the trajectory $H \leftarrow H \cup \left\{y^{(d)}\right\}$;
\STATE Sample $N_{d} \sim \Geo(r_{d})$;
\STATE Call Algorithm \ref{alg:recursive-rMLMC} for $2^{N_{d}}$ times with inputs $\{d + 1; H; \mathcal{S}; r_{d+1},...,r_{D-1}\}$, and label the observations as $R_{d+1}(\yto{d})(1),...,R_{d+1}(\yto{d})\left( 2^{N_{d}} \right)$;
\STATE Calculate $S_{2^{N_d}}, \Se_{2^{N_d - 1}}, \So_{2^{N_d - 1}}$ defined in Section \ref{subsec:notations};
    \STATE Calculate $\left(\text{note } \Delta_0 \coloneqq g_{d}\left( \yto{d}, R_{d+1}(\yto{d})(1) \right) \right)$:
\begin{align*}
    \Delta_{N_{d}} &
=
g_{d}\left( 
\yto{d}, \frac{S_{2^{N_{d}}}}{2^{N_{d}}}
\right)
- \\
&
\frac12
\left[ 
g_{d}\left( 
\yto{d}, \frac{\So_{2^{{N_{d}} - 1}}}{2^{{N_{d}} - 1}}
\right)
+
g_{d}\left( 
\yto{d}, \frac{\Se_{2^{{N_{d}} - 1}}}{2^{{N_{d}} - 1}}
\right)
\right];
\end{align*}
\STATE {\bfseries Return:} 
$R_d \coloneqq \Delta_{N_{d}}/p_{r_{d}}(N_{d})$.
    \ENDIF
\end{algorithmic}
\end{algorithm}

\subsection{Theoretical guarantees}\label{subsec:thoery}
We now discuss the computational costs of Algorithm \ref{alg:recursive-rMLMC} and the statistical properties of $\READ$. Our theoretical results depend on the smoothness conditions of $\{g_d\}_{d=0}^{D-1}$, so we will examine the LBS and LBL cases separately.

\subsubsection{The LBS case} \label{subsubsec: LBS}
The following theorem shows, under the LBS assumption, the computational cost and the variance of $\READ$ can be controlled simultaneously.

\begin{theorem}\label{thm:general D, second order case}
    Suppose for every $d \in \{0,1,\ldots, D-1\}$, the function $g_d$ satisfies the LBS assumption  \eqref{eqn: last-second-derivative}, $r_d \coloneqq 1 - 2^{-k_d}$ satisfies $k_d \in \left( 1, \frac{2^{d+1}}{2^{d+1} - 1} \right)$, and $\lVert g_D(\yto{D})\rVert_{\pi, 2^{D+1}} < \infty$. Then for every $0\leq d \leq D$,   the output $R_d( \yto{d-1})$ of Algorithm \ref{alg:recursive-rMLMC} with inputs \{depth = $d$, trajectory = $\yto{d-1}$, $\calS$, parameters $(r_d, \ldots, r_{D-1})$\} satisfies: 
\begin{itemize}
    \item For $\pi$-almost  every fixed $\yto{d-1}$, 
    $$\bE\left[ R_{d}( \yto{d-1})\mid \yto{d-1}\right] = \gamma_{d}(\yto{d-1}).$$
    \item The expected computational cost of $R_d$ equals
    \begin{align*}
       \prod_{k= d}^{D-1} \frac{r_k}{2r_k - 1} < \infty.
    \end{align*}
    \item The output has finite $2^{d+1}$-th moment, i.e., 
    $$\bE_\pi\left[| R_{d}( \yto{d-1}) |^{2^{d+1}}\right] 
    < \infty~~ \text{for}~~ 0 \leq d \leq D.$$
\end{itemize}
\end{theorem}
Theorem \ref{thm:general D, second order case} states for $\pi$-almost  every $\yto{d-1}$, the expectation of the output $R_d$ conditioning on the input is unbiased for $\gamma_d(\yto{d-1})$. The computational cost has a finite expectation, and the output has a finite $2^{d+1}$-th moment \footnote{Readers should notice that the expectation of $R_d(\yto{d-1})$ is calculated under the conditional distribution $\pi_{d:D}(\cdot \mid \yto{d-1})$. The computational cost and the $2^{d+1}$-th moment are calculated under the joint distribution $\pi$. When the input depth $= 0$, these two underlying distributions coincide.}. The detailed proof of Theorem \ref{thm:general D, second order case} will be provided in the Appendix. Here, we highlight two special cases. First, Theorem \ref{thm:general D, second order case} shows that $\READ$, the output $R_0$ of Algorithm \ref{alg:recursive-rMLMC} when given input \{depth = $0$, trajectory = $\varnothing$, $\calS$, parameters = $(r_0, \ldots, r_{D-1})$\}, has the desired properties. Specifically, it is an unbiased estimator for $\gamma_0$ with finite expected computational cost and finite variance. Second, Theorem \ref{thm:general D, second order case} includes Theorem \ref{thm:nesting_D=1} as a special case when $D = 1$.

Let $R_{0,1}, R_{0,2}, \ldots $ be the i.i.d. outcomes by repeatedly implementing Algorithm \ref{alg:recursive-rMLMC}. Since each $R_{0,i}$ is unbiased and has a finite variance, the standard Central Limit Theorem (CLT) implies that $\sqrt{n} (\sum_{i=1}^n R_{0,i}/n - \gamma_0) \rightarrow \bN(0,1)$ in distribution. This means that the estimator $\sum_{i=1}^n R_{0,i}/n$ converges to $\gamma_0$ at a rate of $n^{-1/2}$ in rMSE (or equivalently, $n^{-1}$ in MSE), which compares quite favorably with the rates obtained by NMC estimators in \cite{rainforth18nesting_prob}. This rate is optimal in the sense that it matches the minimax lower bound over all Monte Carlo methods (Theorem 2.1 of \cite{heinrich1999monte}).  The next corollary shows that, by repeatedly implementing Algorithm \ref{alg:recursive-rMLMC}, users can easily obtain an unbiased estimator for $\gamma_0$ with at most $\varepsilon$-rMSE within $O(\varepsilon^{-2})$ computational cost. 
\begin{corollary}\label{cor: computational cost, LBS case}
With all the assumptions the same as Theorem \ref{thm:general D, second order case}, for any $\varepsilon > 0$, we can construct an estimator $R$ with expected computational cost $\calO(\varepsilon^{-2})$ such that the rMSE $\sqrt{\bE[(R - \gamma_0)^2]}$ is at most $\epsilon$.
\end{corollary}
\begin{proof}[Proof of Corollary \ref{cor: computational cost, LBS case}]
Calling Algorithm \ref{alg:recursive-rMLMC} independently for $n$ times with \{depth = $0$, trajectory = $\varnothing$, $\calS$, parameters = $(r_0, \ldots, r_{D-1})$\} yield i.i.d. unbiased estimators $R_{0,1},...,R_{0,n}$ for $\gamma_0$. Let our estimator be $R \coloneqq \frac1n \sum_{i=1}^n R_{0,i}$. Then,
\[
\bE[(R - \gamma_0)^2]
=
\bE\left[
\left(
\frac1n \sum_{i = 1}^n R_{0,i} - \gamma_0 
\right)^2
\right] 
=
\frac1n \var(R_0).
\]
Thus noting that $\var(R_0) < \infty$ by Theorem \ref{thm:general D, second order case},  taking $n = \var(R_0)/\varepsilon^2$ samples ensures $R$ has up to $\varepsilon$-rMSE. Finally, let $C\coloneqq C(D) < \infty$ be the expected computational cost for one call of Algorithm \ref{alg:recursive-rMLMC} . The expected computational cost for constructing $R$ is then $C \cdot \var(R_0)/\varepsilon^2 = \Theta(\varepsilon^{-2})$. 
\end{proof}
We add  two additional remarks regarding the above corollary. Firstly, while the above result demonstrates that our algorithm achieves optimal dependency on $\epsilon$, it is important to highlight that we are operating within the context of the 'fixed $D$' regime, where the constant in our $\calO$ notation depends on $D$. In fact, it is clear from Theorem \ref{thm:general D, second order case} that each invocation of Algorithm \ref{alg:recursive-rMLMC} has a cost of at least $\Omega((1 + \omega)^D)$ for some $\omega > 0$, indicating that our algorithm does not scale well with increasing nesting levels. Nevertheless,  our algorithm remains practically relevant in scenarios where $D$ is small or moderate, including the examples discussed in Section \ref{sec:intro}. Secondly, the $\varepsilon$-rMSE of $R$ can be easily translated to other performance metrics via standard inequalities. For example, for any $\delta$, Markov's inequality implies the absolute error $\lvert R - \gamma_0\rvert$ is less than $\varepsilon/\sqrt{\delta}$ with probability at least $1-\delta$.

Next, we discuss the assumptions and proof strategies of Theorem \ref{thm:general D, second order case}. We require the first $D$ functions $\{g_d\}_{d=0}^{D-1}$ all satisfy the LBS condition, and the final function $g_D$ has finite $2^{D+1}$-th moment under $\pi$. The LBS assumption also appears in the work of NMC estimators (see the second part of Theorem 3 in \cite{rainforth2018nesting}). The moment assumption of $g_D$ is not required in  \cite{rainforth2018nesting}. Nevertheless, it is a mild assumption that holds in most practical applications. It covers all the cases where $g_D$ is bounded or has a  moment generating function (including the uniform, Gaussian, Poisson, or exponential distributions), which implies $\bE[\lvert g_D\rvert^k] < \infty$ for every $k$. As we will see in our proofs,  these assumptions help us to establish the moment guarantee  of Theorem \ref{thm:general D, second order case} in a backward inductive way. For example, the $2^{D+1}$-th moment assumption on $g_D$ and the LBS assumption on $g_{D-1}$ implies $R_{D-1}$ has finite $2^{D}$-th moment. More generally, the finiteness of the $2^{d+1}$-th moment of $R_d$ follows from the LBS assumption on $g_{d}$ and the $2^{d+2}$-th moment of $R_{d+1}$ (which is the conclusion of the previous inductive step). Eventually, we conclude $R_0$ has a finite variance. Finally, we want to emphasize our moment assumption on $g_D$ is not `trajectory-dependent'. We require $g_D$ has finite $2^{D+1}$-th moment under the joint distribution $\pi$ of $\yto{D}$, which is much weaker than $g_D$ has a uniformly bounded finite  $2^{D+1}$-th moment under $\pi_{D}(\cdot \mid \yto{D-1})$ for every fixed trajectory $\yto{D-1}$.  

Finally, the parameters $\{r_d\}_{d=0}^{D=1}$ reflect the trade-off between the variance and computation cost. Since $2^{N_d}$ calls are required for each $d$,  standard calculation shows that $\bE[2^{N_d}] =  r_d/(2r_d - 1)$ when $r_d > 0.5$, and $+\infty$ if $r_d \leq 0.5$.  Therefore, every $r_d$ has to be strictly greater than $0.5$ to ensure a finite expected computational cost. Meanwhile, we cannot guarantee finite variance or unbiasedness of $\READ$ when $r_d$ becomes too large. Our range for $r_d$ in Theorem \ref{thm:general D, second order case} follows from a careful calculation in our proof to ensure unbiasedness, finite computational cost, and variance simultaneously. 

\subsubsection{The LBL case}\label{subsubsec: LBL}
The  assumptions in Theorem \ref{thm:general D, second order case} guarantee that $\READ$ enjoys an optimal convergence rate and computational cost. However, the second-order derivative assumption also rules out many functions of practical interest, such as $\max$ and $\min$. In this section, we study the theoretical properties of Algorithm \ref{alg:recursive-rMLMC} and $\READ$ under weaker smoothness and moment assumptions. Our result is summarized below:

\begin{theorem}\label{thm:general D, Lipschitz case}
   Fix any $0 <\delta < 1/2$. Suppose for every $d \in \{0,1,\ldots, D-1\}$, the function $g_d$ satisfies the LBL assumption defined in \eqref{eqn: last-lipschitz}, and $r_d \coloneqq 1 - 2^{-k_d}$ satisfies 
   \[
k_d\in\left(1 ,
\left( \frac{2^{d+2} - 3\delta}{2^{d+3} - 3\delta}\right) 
\left( \frac{2^{d+1} - \delta}{2^d - \delta}\right)
\right).
\]
   Moreover, suppose $\lVert g_D(\yto{D})\rVert_{\pi, 2} < \infty$. Then for every $0\leq d \leq D$, the output $R_d( \yto{d-1})$ of Algorithm \ref{alg:recursive-rMLMC} with inputs \{depth = $d$, trajectory = $\yto{d-1}$, $\calS$, parameters $(r_d, \ldots, r_{D-1})$\} has the following properties: 
\begin{itemize}
    \item For $\pi$-almost  every fixed $\yto{d-1}$, 
    $$\bE\left[ R_{d}( \yto{d-1})\mid \yto{d-1}\right] = \gamma_{d}(\yto{d-1}).$$
    \item The expected computational cost of $R_d$ equals
    \begin{align*}
       \prod_{k= d}^{D-1} \frac{r_k}{2r_k - 1} < \infty.
    \end{align*}
    \item The output has finite $(2 - \delta/2^d)$-th moment, i.e., 
    $$\bE_\pi\left[| R_{d}( \yto{d-1}) |^{(2 - \delta/2^d)}\right] 
    < \infty~~ \text{for}~~ 0 \leq d \leq D.$$
\end{itemize}
\end{theorem}

Comparing Theorem \ref{thm:general D, second order case}, which requires the LBS assumption for $\{g_d\}_{d=0}^{D-1}$ and finite $2^{D+1}$-th moment for $g_D$, with Theorem \ref{thm:general D, Lipschitz case}, which only requires the LBL assumption for $\{g_d\}_{d=1}^{D-1}$ and finite second moment for $g_D$, it is clear that Theorem \ref{thm:general D, Lipschitz case} has more general assumptions. However, it does not guarantee that $\READ$ has a finite variance. Nevertheless, it remains unbiased and has a finite expected computational cost. To minimize the loss of moment guarantees, one can choose suitable parameters such that $\READ$ has finite $(2-\delta)$-th moment for any small $\delta$.

Again, let $R_{0,1}, R_{0,2}, \ldots, $ be the i.i.d. outcomes by repeatedly implementing Algorithm \ref{alg:recursive-rMLMC}. There are  more technical challenges when analyzing the convergence rate of $\sum_{i=1}^n R_{0,i}/n - \gamma_0$, as the CLT cannot be applied. Instead, we use the Marcinkiewicz-Zygmund generalized law of large numbers (see Theorem \ref{thm:MZ-LLN} in  Appendix \ref{sec:auxiliary}), which shows $n^{-1}\bE[\lvert \sum_{i=1}^n X_i\rvert^p]\rightarrow 0$ if $\{X_i\}_{i=1}^n$ are i.i.d., centered random variables with finite $p$-th moment for $p\in[1,2)$. Our result is the following: 
\begin{corollary}\label{cor: LBL convergence rate} With all the assumptions the same as Theorem \ref{thm:general D, Lipschitz case}, let $R_{0,1}, R_{0,2}, \ldots, $ be the i.i.d. outcomes by repeatedly implementing Algorithm \ref{alg:recursive-rMLMC}, we have:
    \begin{itemize}
        \item  $\bE[\lvert\sum_{i=1}^n R_{0,i}/n - \gamma_0\rvert] = o(n^{-1/(2(1+\delta))})$.
        \item We can construct an estimator $R$ with expected computational cost $\calO(\varepsilon^{-2(1+\delta)})$ such that the mean absolute error $\bE[\lvert R - \gamma_0 \rvert] < \varepsilon$.
    \end{itemize}
\end{corollary}
\begin{proof}[Proof of Corollary \ref{cor: LBL convergence rate}]
    Applying Theorem \ref{thm:MZ-LLN} with $p = 2-\delta, X_i = R_{0,i} - \gamma_0$ and Jensen's inequality, we have:
    \begin{align*}
        &\bE\left[\left\lvert\sum_{i=1}^n R_{0,i}/n - \gamma_0\right\rvert\right] = n^{-1}\bE\left[\left\lvert \sum_{i=1}^n X_i \right\rvert\right]\\
        &\leq n^{-1} \left(\bE\left[\left\lvert \sum_{i=1}^n X_i \right\rvert^p\right]\right)^{1/p} = o(n^{-1 + \frac 1p}) = o(n^{\frac{-1}{2(1+\delta)}}),
    \end{align*}
    which proves the first part.  The last step follows from \\$1 - 1/(2-\delta) > 1/(2 + 2\delta)$ for $\delta \in (0,1/2)$. 
Setting $n = \Omega(\varepsilon^{-2(1+\delta)}) $ and the second part immediately follows. 
\end{proof}
Although we are not able to recover the optimal $n^{-1/2}$ convergence rate under this more general regime, our convergence rate is still near-optimal as it  can be as close to $n^{-1/2}$ as we want. Still, the convergence rate  does not depend on $D$, and, although we replace the MSE by MAE due to the moment constraint, one can still use Markov's inequality to show the absolute error $\lvert R -\gamma_0 \rvert$ is less than $\varepsilon/\delta$ with probability at least $1 - \delta$.

As the $\max$ function satisfies the LBL assumption, our results here include the optimal stopping problem as a special case. Our results complement the work of \cite{zhou2022unbiased}, where the authors use rMLMC to design an estimator with $\calO(\varepsilon^{-2})$ computational cost under stronger assumptions (see their Assumption 4). We have a slightly worse cost of $\calO(\varepsilon^{-2(1+\delta)})$ but under more general assumptions.
\section{Numerical experiments}\label{sec:numerical}
We test our algorithm on three examples. Some additional statistics and an extra experiment are provided in Appendix \ref{sec: additional statistics} and \ref{sec:extra experiment}. Our code is available at \url{https://github.com/guanyangwang/rMLMC_RNE}.

\subsection{A toy example}
We consider the following simple example with known ground-truth. Suppose the process $(\y0,\y1,\y2)$ satisfies $\y0\sim \bN(\pi/2,1), \y1 \sim \bN(\y0,1), \y2\sim \bN(\y1,1)$. Define $g_0(\y0, z) \coloneqq \sin(\y0+z), g_1(\yto1, z) \coloneqq \sin(\y1 - z)$, and $g_2(\yto2) \coloneqq \y2$. The target quantity $\gamma_0$ defined \eqref{eqn:nested target} is a nested expectation with  $D = 2$. One can use the formula $\bE_{Z\sim \bN(\mu,\sigma^2)}[\sin(Z)] = \sin(\mu) \exp(-\sigma^2/2)$ to analytically calculate $\gamma_0 = \exp(-1/2)\approx 0.6065$. Now we compare our $\READ$ estimator with the NMC estimators in \cite{rainforth2018nesting}. 

For the NMC estimator, users first specify $N_0, N_1, N_2$. Then we sample $N_0$ copies of $\y0$, $N_1$ copies of $\y1$ for each fixed $\y0$, and $N_2$ copies of $\y2$ for each fixed $\yto1$, and use these samples to form the NMC estimator, details are explained in Appendix \ref{sec: NMC}. Following \cite{rainforth18nesting_prob}, we consider two ways of allocating $(N_0, N_1, N_2)$. The first estimator NMC1 is to choose $N_0 = N_1 = N_2$, the second NMC2 is to choose $N_0 = N_1^2 = N_2^2$. For $\READ$, since all assumptions in Theorem \ref{thm:general D, second order case} are satisfied, therefore when $r_0\in (1/2, 3/4)$ and $r_1 \in (1/2, 1- 2^{-4/3})$, the $\READ$ estimator generated by Algorithm \ref{alg:recursive-rMLMC} is unbiased and of finite variance. Since the computational cost gets lower when each $r_i$ gets larger, we choose $r_0 = 0.74$ and $r_1 = 0.6$ (close to the upper-end of their respective ranges above) to facilitate the computational efficiency. Therefore, implementing Algorithm \ref{alg:recursive-rMLMC} once has an expected sample size/computational cost $\left(r_1/(2r_1 - 1)\right) \left(r_2/(2r_2 - 1)\right) \approx 4.625$. 


Our comparison result is summarized in Figure \ref{fig:error_comparison}. Since the NMC methods and $\READ$ have different ways of generating estimators. To make a fair comparison, we compare the estimation errors with the total sample cost used by these three estimators. For NMC1 and NMC2, the total sample cost is $n = N_0 N_1 N_2$. For $\READ$, the total sample cost is random, therefore we use its expected value, which equals $4.625 \times$ Number  of repetitions of Algorithm \ref{alg:recursive-rMLMC}. The slopes of the blue, red, and green lines, which correspond to the empirical convergence rate of $\READ$, NMC1, NMC2, equals $-0.97, -0.35, -0.47$, respectively. They match well with the theoretical predictions $n^{-1}$ in Corollary \ref{cor: computational cost, LBS case} for $\READ$, $n^{-1/3}$ for NMC1, and $n^{-1/2}$ for NMC2  in \cite{rainforth2018nesting}.  It is clear from Figure \ref{fig:error_comparison}(a) that  READ  has a significant advantage over NMC estimators, with both faster convergence rate and orders of magnitude lower errors.

\begin{figure}
    \centering
    \subfigure[]{\includegraphics[width=0.24\textwidth]{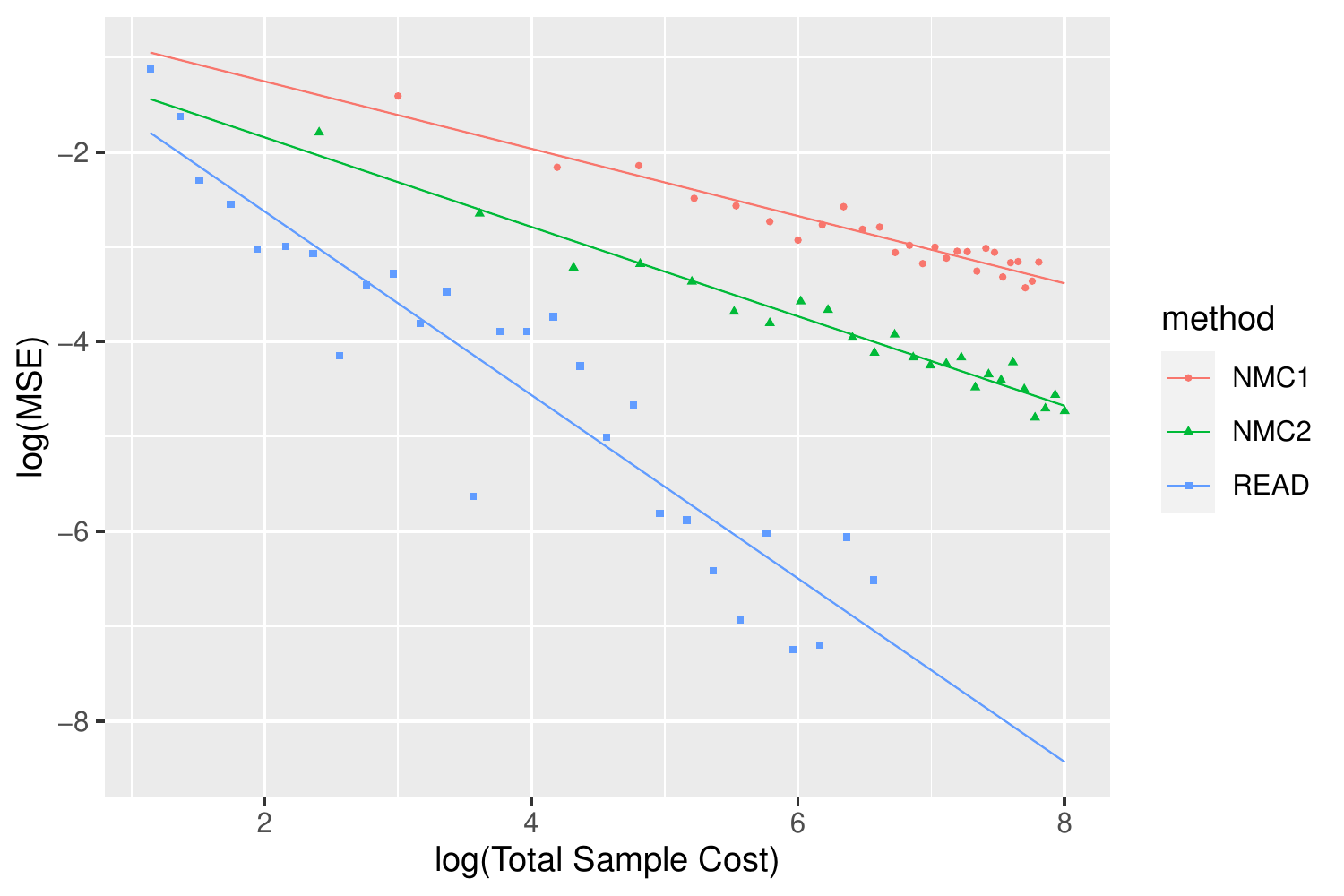}} 
    \subfigure[]{\includegraphics[width=0.23\textwidth]{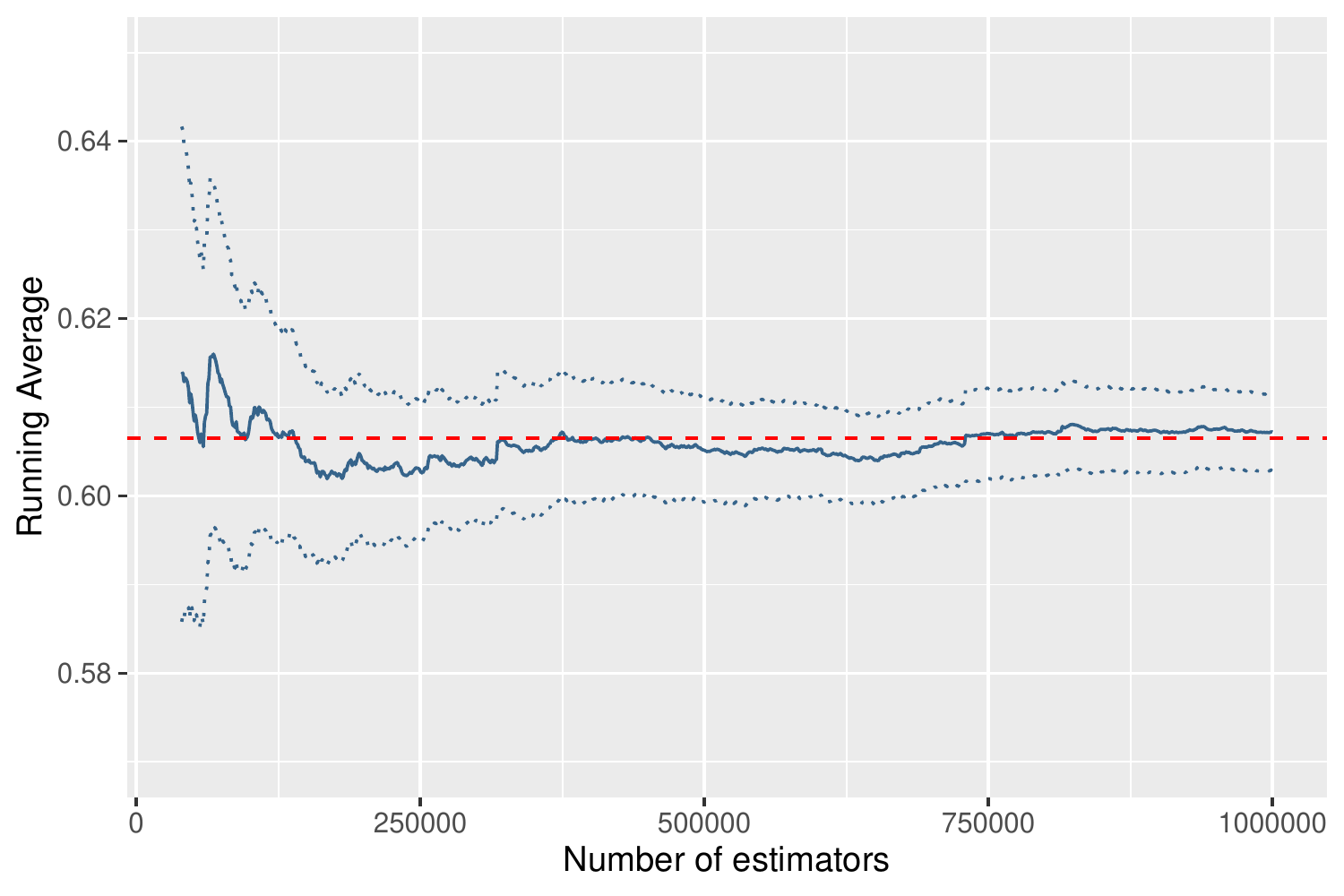}} 
    \caption{(a): The comparison on the empirical MSEs of estimating the RNE among $\READ$ (blue), NMC1 (red), and NMC2 (green). All the logarithms are of the base $10$. Each method's empirical errors are calculated based on $20$ independent repetitions. (b) The trace plot (solid blue curve) of the running averages of $\READ$.  The blue dotted curves are the $95\%$ confidence intervals. The red dashed line is the ground truth $\exp(-1/2)$.}
    \label{fig:error_comparison}
\end{figure}

We also call Algorithm \ref{alg:recursive-rMLMC} for $10^6$ times and plot the running averages of our estimates in Figure \ref{fig:error_comparison}(b). Our estimator becomes more accurate when we increase the number of repetitions. For each $k\in (1,2,\ldots, 10^6)$, we also calculate the standard deviation (SD) of the first $k$ repetitions and use Mean $\pm 1.96$ SD to form the $95\%$ confidence interval. It is also clear from Figure \ref{fig:error_comparison}(b) that our confidence intervals always include the ground-truth, suggesting the high accuracy of our method. In contrast, constructing confidence intervals of NMC estimators are much more time-consuming.

\subsection{Example with heavy-tail underlying distribution}
All three estimators are also evaluated on the same set of functions using an independent, non-central $t$-distribution with $10$ degrees of freedom and a noncentrality parameter of $0.5$. The outcomes of these tests are illustrated in Figure \ref{fig:comparison_t}. Despite the fact that the $t$-distribution exhibits a significantly heavier tail compared to the Gaussian distribution, it is evident from Figure \ref{fig:comparison_t}(a) that the convergence rate, as indicated by the speed at which each color converges to the black dotted line, is considerably faster for the $\READ$ method compared to the NMC estimators.
\begin{figure}
    \centering
    \subfigure[]{\includegraphics[width=0.24\textwidth]{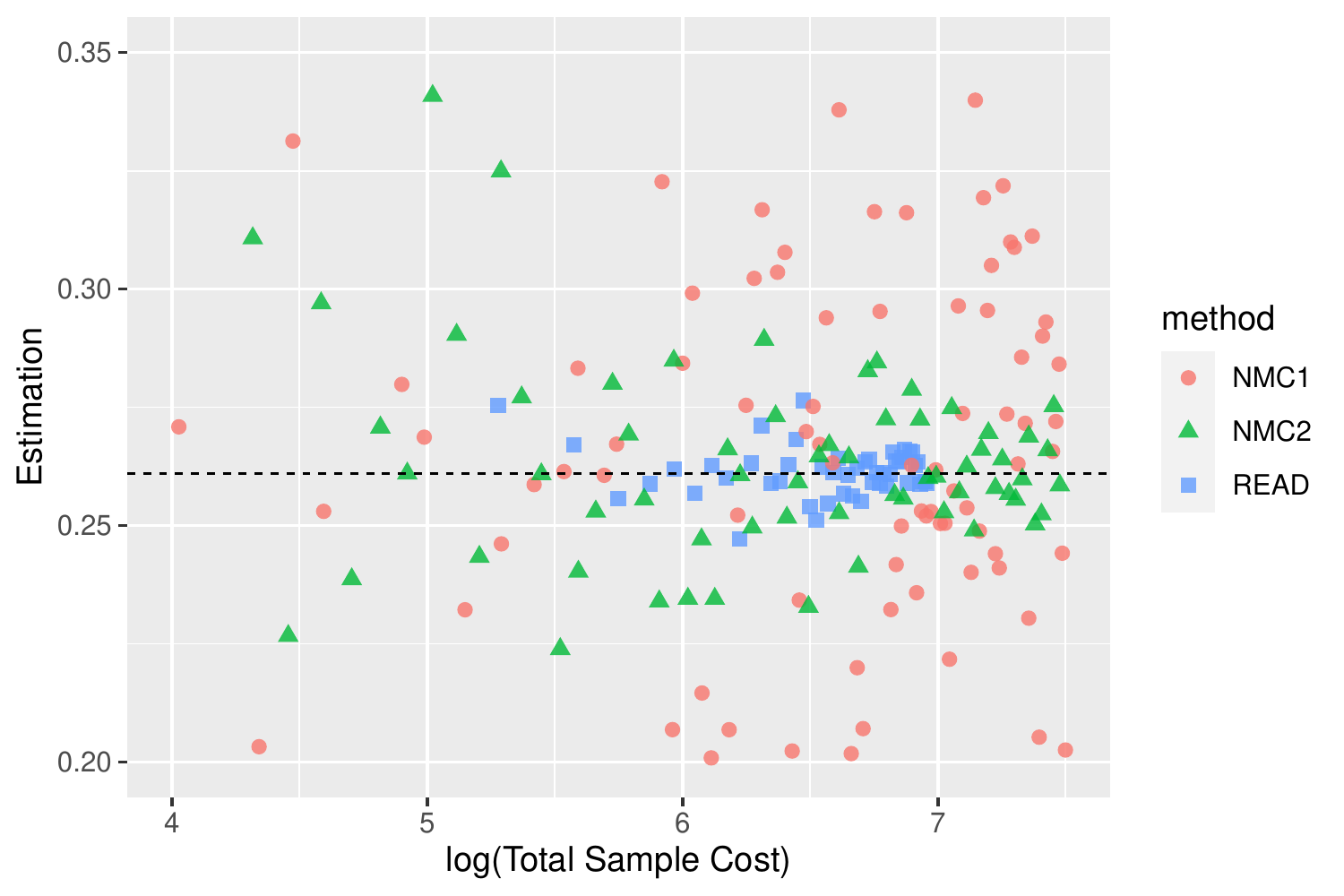}} 
    \subfigure[]{\includegraphics[width=0.23\textwidth]{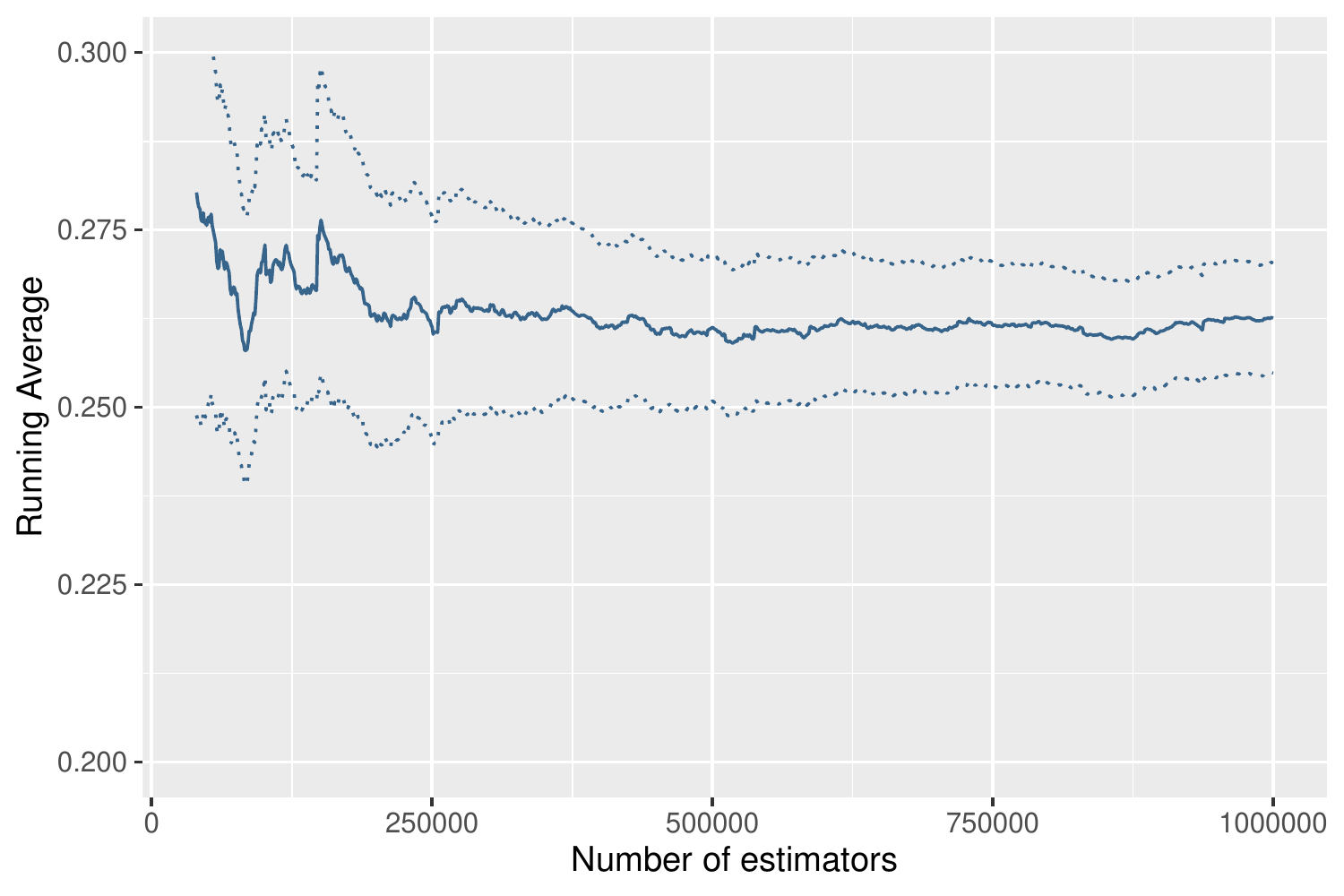}} 
    \caption{(a): Scatterplot of the estimation of $\gamma_0$. Blue, red, green points correspond to $\READ$, NMC1, NMC2 estimators respectively. (b) Trace plot (solid blue curve) of the running averages of $\READ$.  The blue dotted curves are the $95\%$ confidence intervals.  }
    \label{fig:comparison_t}
\end{figure}
\subsection{Pricing the Bermudan Options}
Finally, we utilize our method to price high-dimensional Bermudan basket put options.
 Given that option pricing can be formulated as an optimal stopping problem, our estimator simplifies to the MUSE estimator in \cite{zhou2022unbiased}. The underlying process $y^{(0:D)} = (S_0, S_{T/D}, S_{2T/D}, \ldots, S_{T})$ where $S_t$ is a $M$-dimensional geometric Brownion motion, each coordinate follows $d{S_i(t)} = (r-\delta)S_i(t) dt + \sigma S_i(t) dW_i(t).$ For $d\leq D-1$, our $g_d (y^{(0:d)}, z) := \max\{e^{-rT/D}, z\}$, where $U(x) := \max\{(K - \overline{x}), 0\}$. For $d = D$, $g_D(y^{(0:D)}) := U(y^{(D)})$. We also adopt the standard parameters in \cite{jain2012pricing, bender2006policy,zhou2022unbiased}: $T = 3, M = 5,\sigma = 0.2, r = 0.05, K = y^{(0)}_i = 100$ for every $i.$ 

 We follow previous works and set $D = 3$ . We tested $\READ$ on a $500$-core cluster, with each computer generating $10^4$ estimators. Our estimator is obtained by averaging all the $5$ million estimators. We also tested NMC1 and NMC2 on the same cluster. Each computer generates $10$ estimators for both methods. NMC1 uses $N_0 = N_1 = \cdots = 80$, NMC2 uses $N_0 = N_1^2 = .. = 900$. The results are summarized in Table \ref{tab:option, D3M5}. After comparing with existing algorithms tailored for option pricing/optimal stopping, we observe that $\READ$ aligns closely with the results of previous works. In contrast, both NMC1 and NMC2 exhibit a significant overestimation of the target. This discrepancy arises from the convex nature of the max function, which introduces a systematic bias in the NMC estimators. Therefore, our method provides more reliable estimates with a much shorter completion time. Similarly, we test the case where $D = 4$. $\READ$  generates $5$ million estimators over $500$ processors. NMC1 generates $5000$ estimators with $N_0 = 25$, NMC2 generates $5000$ estimators with $N_0 = 225$. The outcomes of these tests are also presented in Table \ref{tab:option, D3M5}.   Again, both NMC estimators  overestimate the target, albeit a smaller standard error.  
\begin{table}[htbp!]
\centering
  \resizebox{0.45\textwidth}{!}{%
\begin{tabular}{|c|c|c|c|}\hline
$D = 3$                    & $\READ$              & NMC1                    & NMC2                    \\\hline
Cost                 & $6.9\times 10^7$     & $2.05\times 10^{11}$          & $1.22\times 10^{11}$      \\
Time/s  & $(14.6, 22.6, 87.3)$ & $(116.1, 169.1, 215.9)$ & $(316.3, 350.7, 377.8)$ \\
Estimate (se)                     & $2.159 (0.008)$      & $2.169 (0.005)$         & $2.180 (0.014)$        \\\hline
$D = 4$                   & $\READ$              & NMC1                    & NMC2                    \\\hline
Cost                 & $2.36\times 10^8$     & $4.88\times 10^{10}$          & $5.69\times 10^{10}$      \\
Time/s  & $(19.8, 44.5, 299.8)$ & $(66.6, 118.2, 189.1)$ & $(316.3, 350.7, 377.8)$ \\
Estimate (se)                     & $2.284 (0.065)$      & $2.357 (0.008)$         & $2.393 (0.003)$ \\\hline
\end{tabular}%
}
\caption{Summary of results when $D = 3$. The three values in the ``Time" row correspond to the minimum, average, and maximum completion times across 500 processors.}
  \label{tab:option, D3M5}
\end{table}

\section{Further discussions}\label{sec:discussion}
Here we provide some remarks for practical implementation and discuss some potential generalizations. The users need to specify the parameters $\{r_d\}_{d=0}^{D-1}$ when implementing Algorithm \ref{alg:recursive-rMLMC}. Larger values of $r_i$ lead to a shorter time for each implementation but potentially larger variance. When some $r_i$ is not chosen according to Theorem \ref{thm:general D, second order case} or \ref{thm:general D, Lipschitz case}, the algorithm can still be implemented, but the variance may be infinite. The trade-off between the values of $\{r_d\}$ and the fluctuations of the resulting estimator is problem-specific. In practice, knowing how many repetitions are sufficient is important to provide an accurate estimator. One possible way is to bound certain moments of $\READ$ and use Corollary \ref{cor: computational cost, LBS case} or \ref{cor: LBL convergence rate} to choose a sufficiently large $n$. But this bound can be problem-specific and very conservative. Instead, we follow \cite{glynn2014exact} and suggest the following adaptive stopping rule: users first specify a precision-level $\varepsilon$ and a small $\delta\%$. When repeatedly implementing Algorithm \ref{alg:recursive-rMLMC}, users calculate the empirical $(1-\delta\%)$ confidence interval $[L_\delta(k), U_\delta(k)]$ for first $k$ repetitions in the same way as Section \ref{sec:numerical} for every $k$. Users can stop when the width of the confidence interval is less than $2\varepsilon$. The validity of this stopping rule is proven in \cite{glynn1992asymptotic}. 

One potential direction for extension is as follows. Here we only consider the `fix $D$' regime and construct estimators with optimality guarantees. However, the cost of Algorithm 1 scales exponentially with $D$. Therefore, although our algorithm is more efficient than the NMC estimator for every fixed $D$, both methods are not practical when $D$ becomes too  large. Indeed, the poor scaling with $D$ is a common issue in related literature such as \cite{glasserman2004number,zanger2013quantitative} and seems unavoidable. An interesting direction would be to construct  modifications of Algorithm \ref{alg:recursive-rMLMC} under extra practical assumptions for large or infinite $D$. For example, if we know that the `influence' of $\gamma_d$ on $\gamma_0$ decays exponentially or double-exponentially with $d$, it is then  sufficient to truncate the depth to $\tilde D \coloneqq \log(1/\varepsilon)$ or $\sqrt{\log(1/\varepsilon)}$. We hope to report progress in the future.
 
\section*{Acknowledgements}
The authors would like to thank Tom Rainforth, Takashi Goda, Pierre Jacob, and three referees for their helpful comments. Guanyang Wang gratefully acknowledges support by the National Science Foundation through grant  DMS-2210849 and the Adobe Data Science Research Award.

\bibliography{ref}
\bibliographystyle{icml2023}

\newpage
\appendix
\onecolumn


\section{Auxiliary results}\label{sec:auxiliary}
The following theorem is from pg. 150 \citep{Gut2005-GUTPAG}.
\begin{theorem}\label{indep-MZ-ineq}
Let $p \geq 1$. Suppose that $X_1, X_2, ... X_n$ are independent random variables, with mean $0$ and $\bE|X_k|^p < \infty$ for all $k$, and let $S_n\coloneqq \sum_{i=1}^n X_i$ denote the partial sums. Then there exist constants $A_p^*, B_p^*$ depending only on $p$ such that
$$
A_p^*\bE\left( 
\sum_{k=1}^n X_k^2
\right)^{p/2}
\leq 
\bE|S_n|^p
\leq
B_p^*\bE\left( 
\sum_{k=1}^n X_k^2
\right)^{p/2}.
$$
\end{theorem}

The following corollary to the above theorem is from pg. 151 \citep{Gut2005-GUTPAG}, Corollary 8.2. 
\begin{corollary}\label{cor:iid-MZ-ineq}
Let $p \geq 1$. Suppose that $X_1, X_2, \ldots$ are i.i.d. random variables, with mean $0$ and $\bE|X_1|^p < \infty$, and let  $S_n\coloneqq \sum_{i=1}^n X_i$ denote the partial sums. Then there exists a constant $B_p$ depending only on $p$, such that
$$
\bE|S_n|^p
\leq
\begin{cases}
B_pn^{p/2}\bE|X_1|^{p}, & p > 2 \\
B_pn\bE|X_1|^{p}, & 1 \leq p \leq 2.
\end{cases}
$$
\end{corollary}

The following lemma is instrumental in the proofs for the theoretical guarantees of our algorithm under both the LBS and LBL assumptions.

\begin{lemma}\label{lem:conditional-MZ}
Let $(Z_1,Z_2)$ be a 2-stage stochastic process and there exists $p \geq 1$, such that $ \bE[|Z_2|^p] < \infty$. Conditioning on $Z_1$, sample $i.i.d.$ $Z_2(1),...,Z_2(n)$. Then,
\[
\bE
\left[\left|
\frac1n \sum_{i=1}^n Z_2(i)
-
\bE[Z_2 \mid Z_1]
\right|^p
\right]
\leq
\begin{cases}
B_p'
\frac{\bE[|Z_2|^p]}{n^{p/2}} & p > 2 \\

B_p'
\frac{\bE[|Z_2|^p]}{n^{p - 1}} & 1 \leq p \leq 2
\end{cases}
\]
\end{lemma}
\begin{proof}
Let $p > 2$. For arbitrary fixed $Z_1 = z_1$,  define $\bar Z_2(i)\coloneqq Z_2(i) - \bE[Z_2(i)\mid Z_1 = z_1]$, and apply Corollary \ref{cor:iid-MZ-ineq} on the $i.i.d.$ mean $0$ random variables $\bar Z_2(i) $ under the probability distribution $\pi(\cdot \mid Z_1 = z_1)$, we have 
\begin{align*}
\bE
\left[ 
\left| 
\frac1n \sum_{i=1}^n Z_2(i) - \bE[Z_2 \mid Z_1]
\right|^p
\right] 
&=
\int_\Omega 
\bE
\left[ 
\left| 
\frac1n \sum_{i=1}^n Z_2(i) - \bE[Z_2 \mid Z_1]
\right|^p
\mid 
Z_1 = z_1
\right] 
\pi_1(\df z_1) \\
& = \frac{1}{n^p} \bE[\left| \sum_{i=1}^n \bar Z_2(i)\right|^p \mid Z_1 = z_1] \pi_1(\df z_1)\\
& \leq \frac{B_p}{n^{p/2}}\bE[\lvert \bar Z_2(1)\rvert^p \mid Z_1 = z_1] \pi_1(\df z_1) \\
& \leq \frac{B_p'}{n^{p/2}} \bE[\lvert  Z_2(1)\rvert^p \mid Z_1 = z_1] \pi_1(\df z_1) \\
& = \frac{B_p'}{n^{p/2}} \bE[\lvert  Z_2(1)\rvert^p\rvert].
\end{align*}

The second inequality follows from the inequality $(a+b)^p\leq 2^{p-1}(\lvert a \rvert^p + \lvert b \rvert^p)$ and the monotonicity of a random variable's $L^p$ norm :
\begin{align*}
    \bE[|X- \bE[X]|^p] \leq 2^{p-1} (\bE[|X|^p] + \lvert\bE[X]\rvert^p ) \leq 2^{p} \bE[\lvert X\rvert^p]
\end{align*}

For the case $1 \leq p \leq 2$, the calculation is identical to the above, except replace the $B_p'/n^{p/2}$ with $B_p'/n^{p-1}$ from Corollary \ref{cor:iid-MZ-ineq}.
\end{proof}

The following theorem is the Marcinkiewicz-Zygmund law of large numbers from pg. 311 \citep{Gut2005-GUTPAG}, which gives us the $\calO\left(\varepsilon^{-2(1 + \delta)}\right)$ sampling complexity for the LBL case for  $0 < \delta < 1/2$. 
\begin{theorem}\label{thm:MZ-LLN}
Suppose that $X_1, X_2,...$ are i.i.d. random variables., and set $S_n = \sum_{k=1}^n X_k, n \geq 1$. If $\bE|X_1|^p < \infty$ and $\bE [X_1] = 0$ when $1 \leq p < 2$, then
\[
\bE 
\left| 
\frac{S_n}{n^{1/p}}
\right|^p
=
\bE 
\frac{|S_n|^p}{n}
\stackrel{n \to \infty}{\longrightarrow}
0.
\]
\end{theorem}

\newpage
\section{Proof of Theorem \ref{thm:general D, second order case}}

Recall that we are considering an $M$-dimensional process $\yto{D}$, i.e. $\y{d} \in \bR^M$ for each $0 \leq d \leq D$. Explicitly, for $d \in \{0,...,D-1\}$, $g_d : \bR^{(d+1)M+1} \to \bR$  and $g_D : \bR^{(D+1)M} \to \bR$.

Then, we say $\{g_d\}_{d=0}^{D-1}$ satisfy the last-component bounded second derivative condition (LBS) if for $z \in \bR$:
\[
\sup_{(\yto{d}, z)} 
\left|\partial^2_{(d+1)M + 1} g_d(\yto{d}, z)\right| < K_d.
\]

\begin{proof}
\textbf{Case 1:~~}$d = D$ \\

When $d = D$, Algorithm \ref{alg:recursive-rMLMC} samples one $y^{(D)} \sim \pi_{D}$ and outputs $R_D(\yto{D-1}) := g_D(\yto{D})$. We first prove our output $R_D(\yto{D-1})$ has a finite expectation for almost  every fixed $\yto{D-1}$, then its expectation equals $\gamma_D(\yto{D-1})$ follows directly from the algorithm design. To show the first point, notice that the  expectation of $|R_D(\yto{D-1})|$ given $\yto{D-1}$  equals the conditional expectation $\bE[\lvert g_D(\yto{D})\rvert\mid \yto{D-1}]$. Since $\bE[|g_D|] < \infty$ by assumption, we have $\bE[\lvert g_D(\yto{D})\rvert\mid \yto{D-1}] < \infty$, almost surely. Therefore Algorithm \ref{alg:recursive-rMLMC} is unbiased when $d = D$ for almost  every input $\yto{D-1}$. Furthermore, it has computational cost $1$, and the output has finite $2^{D+1}$-st moment.

\textbf{Case 2:~~}$0\leq d \leq D-1$\\

Now that our base case is proven, we proceed via backwards induction. Suppose unbiasedness, finite $2^{d+2}$-th moment, and finite expected computational cost are all satisfied for $d+1$ where $0 \leq d \leq D - 1$. Conditioning on $\yto{d-1}$, we sample $y^{(d)} \sim \pi_d$ and $N_{d} \sim \Geo(r_{d})$.  Algorithm \ref{alg:recursive-rMLMC} will call itself independently for $2^{N_{d}}$ times, each with input \{Depth index: $d+1$, Trajectory History: $H = \yto{d}$, Parameters: $r_{d+1}, \cdots, r_{D-1}$\}. This gives us i.i.d. samples $R_{d+1}(\yto{d})(1),...,R_{d+1}(\yto{d})(2^{N_{d}})$ which are used to compute the following:
\begin{align*}
S_{2^{N_{d}}}
&=
R_{d+1}(\yto{d})(1) +
R_{d+1}(\yto{d})(2) +
\cdots +
R_{d+1}(\yto{d})(2^{N_{d}}), \\
\So_{2^{N_{d} - 1}}
&=
R_{d+1}(\yto{d})(1) +
R_{d+1}(\yto{d})(3) +
\cdots +
R_{d+1}(\yto{d})(2^{N_{d}} - 1), \\
\Se_{2^{N_{d} - 1}}
&=
R_{d+1}(\yto{d})(2) +
R_{d+1}(\yto{d})(4) +
\cdots +
R_{d+1}(\yto{d})(2^{N_{d}}).
\end{align*}

Then Algorithm \ref{alg:recursive-rMLMC} returns as output $R_{d} = \Delta_{N_{d}} / p_{r_d}(N_{d})$, where $\Delta_{N_{d}}$ is the antithetic quantity in Algorithm \ref{alg:recursive-rMLMC}. By the inductive hypothesis on $d+1$, we have for almost  every $\yto{d}$:
\[
\bE_{\pi_{d+1 : D}}
[
R_{d + 1} \mid \yto{d}
]
=
\gamma_{d + 1}(\yto{d}),
\quad
\bE_{\pi}\left[ |R_{d+1}|^{2^{d+2}} \right] 
< 
\left(
\prod_{i = d + 1}^{D} \tilde{C}_i
\right)
\norm{g_D(\yto{D})}_{\pi, 2^{D+1}}^{2^{D+1}}.
\]
We will start with showing  $R_d$ has a finite computational cost and a finite $2^{d+1}$-th moment, and then show the unbiasedness. \\
\underline{Finite cost:}\\
To show the finite expected computational cost, recall that implementing Algorithm \ref{alg:recursive-rMLMC} with input depth $d$ requires $2^{N_d}$ calls of Algorithm \ref{alg:recursive-rMLMC} with input depth $d+1$. Since $N_d\sim \Geo(r_d)$ with $r_d > 0.5$, calling Algorithm \ref{alg:recursive-rMLMC} with input depth $d$ has an expected cost:
\begin{align*}
    \frac{r_d}{2r_d - 1} \times  \text{the expected cost of Algorithm \ref{alg:recursive-rMLMC} with input depth}~ d+1,
\end{align*}
where $\frac{r_d}{2r_d - 1} = \bE[2^{N_d}] < \infty$. By our inductive hypothesis, the second term in the above product is finite, therefore the expected cost of Algorithm \ref{alg:recursive-rMLMC} with input depth $d$ is also finite.

\underline{Finite $2^{d+1}$-th moment:}

Next we show $R_{d}$ has a finite $2^{d+1}$-th moment. For every fixed positive integer $n$, doing a Taylor expansions for $g_d$ at $(\yto{d}, \gamma_{d+1})$ with respect to the last component gives us:
\begin{align*}
g_{d}\left( 
\yto{d}, \frac{S_{2^{n}}}{2^{n}}
\right)
&=
g_d(\yto{d}, \gamma_{d+1})
+
\partial_{(d+1)M+1}g_d(\yto{d}, \gamma_{d+1})\left( \frac{S_{2^{n}}}{2^{n}} -  \gamma_{d+1} \right) \\
&\quad+
\frac12 \partial_{(d+1)M+1}^2g_d(\yto{d}, \xi(n))\left(\frac{S_{2^{n}}}{2^{n}} -  \gamma_{d+1} \right)^2 \\
g_{d}\left( 
\yto{d}, \frac{\So_{2^{n - 1}}}{2^{n - 1}}
\right)
&=
g_d(\yto{d}, \gamma_{d+1})
+
\partial_{(d+1)M+1}g_d(\yto{d}, \gamma_{d+1})\left( \frac{\So_{2^{n - 1}}}{2^{n - 1}} - \gamma_{d+1} \right) \\
&\quad+
\frac12 \partial_{(d+1)M+1}^2g_d(\yto{d}, \xio(n - 1))\left( \frac{\So_{2^{n - 1}}}{2^{n - 1}} - \gamma_{d+1} \right)^2 \\
g_{d}\left( 
\yto{d}, \frac{\Se_{2^{n - 1}}}{2^{n - 1}}
\right)
&=
g_d(\yto{d}, \gamma_{d+1})
+
\partial_{(d+1)M+1}g_d(\yto{d}, \gamma_{d+1})\left( \frac{\Se_{2^{n - 1}}}{2^{n - 1}} - \gamma_{d+1} \right) \\
&\quad+
\frac12 \partial_{(d+1)M+1}^2g_d(\yto{d}, \xie(n - 1))\left( \frac{\Se_{2^{n - 1}}}{2^{n - 1}} - \gamma_{d+1} \right)^2,
\end{align*}

with $\xi(n)$ between $\gamma_{d+1}$ and $S_{2^{n}} / 2^{n}$,  $\xio(n - 1)$ between $\gamma_{d+1}$ and $\So_{2^{n - 1}} / 2^{n - 1}$, $\xie(n - 1)$ between $\gamma_{d+1}$ and $\Se_{2^{n - 1}} / 2^{n - 1}$.

Thus, we have:
\begin{equation*}
\begin{split}
\Delta_{n}
&={}
g_{d}\left( 
\yto{d}, \frac{S_{2^{n}}}{2^{n}}
\right)
-
\frac12
\left[ 
g_{d}\left( 
\yto{d}, \frac{\So_{2^{{n} - 1}}}{2^{{n} - 1}}
\right)
+
g_{d}\left( 
\yto{d}, \frac{\Se_{2^{{n} - 1}}}{2^{{n} - 1}}
\right)
\right] \\
&= 
\partial_{(d+1)M+1}g_d(\yto{d}, \gamma_{d+1})\left( \frac{S_{2^{n}}}{2^{n}} - \gamma_{d+1} \right)
+
\frac12 \partial_{(d+1)M+1}^2g_d(\yto{d}, \xi(n))\left( \frac{S_{2^{n}}}{2^{n}} - \gamma_{d+1} \right)^2 \\
&\quad-
\frac12
\Bigg[ 
\partial_{(d+1)M+1}g_d(\yto{d}, \gamma_{d+1})\left( \frac{\So_{2^{n - 1}}}{2^{n - 1}} - \gamma_{d+1} \right) +
\frac12 \partial_{(d+1)M+1}^2g_d(\yto{d}, \xio(n - 1))\left( \frac{\So_{2^{n - 1}}}{2^{n - 1}} - \gamma_{d+1} \right)^2 \\
&\quad\quad\quad+
\partial_{(d+1)M+1}g_d(\yto{d}, \gamma_{d+1})\left( \frac{\Se_{2^{n - 1}}}{2^{n - 1}} - \gamma_{d+1} \right)
+
\frac12 \partial_{(d+1)M+1}^2g_d(\yto{d}, \xie(n - 1))\left( \frac{\Se_{2^{n - 1}}}{2^{n - 1}} - \gamma_{d+1} \right)^2
\Bigg] \\
&=
\frac12 \partial_{(d+1)M+1}^2g_d(\yto{d}, \xi(n))\left( \frac{S_{2^{n}}}{2^{n}} - \gamma_{d+1} \right)^2
-
\frac12
\Bigg[ 
\frac12 \partial_{(d+1)M+1}^2g_d(\yto{d}, \xio(n - 1))\left( \frac{\So_{2^{n - 1}}}{2^{n - 1}} - \gamma_{d+1} \right)^2 \\
&\quad\quad\quad\quad\quad\quad\quad\quad
\quad\quad\quad\quad\quad\quad\quad\quad
\quad\quad\quad\quad\quad\quad\quad\quad+
\frac12 \partial_{(d+1)M+1}^2g_d(\yto{d}, \xie(n - 1))\left( \frac{\Se_{2^{n - 1}}}{2^{n - 1}} - \gamma_{d+1} \right)^2
\Bigg].
\end{split}
\end{equation*}

By the LBS assumption which assumes $\lvert \partial_{(d+1)M+1}^2g_d(\yto{d}, z))\rvert < K_d$ for every $(\yto{d}, z)$, and our inductive hypothesis: $\gamma_{d+1} 
= \bE[R_{d+1}(\yto{d}) \mid \yto{d}]$ which allows us to use Lemma \ref{lem:conditional-MZ} with $Z_1 = \yto{d}$, $Z_2 = R_{d+1}(\yto{d})$, we have:
\begin{align*}
\norm{
\Delta_{n}
}_{\pi, 2^{d + 1}}
&\leq 
K_{d}
\norm{\left(  \frac{S_{2^n}}{2^n} - \gamma_{d+1}\right)^2}_{\pi, 2^{d + 1}}
+
\frac{K_{d}}2
\norm{\left( \frac{\So_{2^{n - 1}}}{2^{n - 1}} - \gamma_{d+1}  \right)^2}_{\pi, 2^{d + 1}}
+
\frac{K_{d}}2
\norm{\left( \frac{\Se_{2^{n - 1}}}{2^{n - 1}} - \gamma_{d+1} \right)^2}_{\pi, 2^{d + 1}} \\
&\leq
K_{d}B'_{2^{d+2}}\left(
\frac{\bE_{\pi}\left[ |R_{d+1}(\yto{d})|^{2^{d+2}} \right]}{2^{(2^{d+1}n)}}
\right)^{1/2^{d+1}} 
+ 
K_{d}B'_{2^{d+2}}\left(
\frac{\bE_{\pi}\left[ |R_{d+1}(\yto{d})|^{2^{d+2}} \right]}{2^{(2^{d+1}(n-1))}}
\right)^{1/2^{d+1}} \\
&=
K_{d}B'_{2^{d+2}}\left(
\frac{\bE_{\pi}\left[ |R_{d+1}(\yto{d})|^{2^{d+2}} \right]}{2^{(2^{d+1}n)}}
\right)^{1/2^{d+1}}
+ 
2K_{d}B'_{2^{d+2}}\left(
\frac{\bE_{\pi}\left[ |R_{d+1}(\yto{d})|^{2^{d+2}} \right]}{2^{(2^{d+1}n)}}
\right)^{1/2^{d+1}} \\
& = 
3K_{d}B'_{2^{d+2}}
\left(
\frac{\bE_{\pi}\left[ |R_{d+1}(\yto{d})|^{2^{d+2}} \right]}{2^{(2^{d+1}n)}}
\right)^{1/2^{d+1}}.
\end{align*}

Therefore, in total:
\[
\bE_{\pi}\left[ |\Delta_n|^{2^{d + 1}} \right]
\leq
D_d
\frac{\bE_{\pi}\left[ |R_{d+1}(\yto{d})|^{2^{d+2}} \right]}{2^{(2^{d+1}n)}},
\]
with $D_d = (3K_{d}B'_{2^{d+2}})^{(2^{d+1})} $.

The result claimed in (c) is obtained as follows. We have for $(1 - r_d) = 2^{-k_d}$ for some $k_d \in \left( 1 , \frac{2^{d+1}}{2^{d+1} - 1} \right)$:
\begin{align*}
\bE_{\pi}
\left[ 
|R_d(\yto{d-1})|^{2^{d+1}}
\right]
&=
\sum_{n=0}^\infty
\frac{\bE_{\pi}\left[ |\Delta_n|^{2^{d+1}} \right]}{p_{r_d}(n)^{{2^{d+1}} - 1}} \\
&\leq
\frac{D_d \bE_{\pi}\left[ |R_{d+1}(\yto{d})|^{2^{d+2}} \right]}{r_d^{2^{d+1} - 1}}
\sum_{n=0}^\infty 
\frac{1}{2^{2^{d+1} n}(1 - r_d)^{({2^{d+1} - 1}) n}} \\
&\leq
C_d
\left(
\prod_{i = d+1}^{D}
\tilde{C}_i
\right) \norm{g_D(\yto{D})}_{\pi, 2^{D+1}}^{2^{D+1}} 
\sum_{n=0}^\infty 
\left( \frac{1}{2^{{2^{d+1}} - k_d(2^{d+1} - 1)}} \right)^n \quad \textcolor{blue}{\text{inductive hypothesis}} \\
&=
C_d
\left(
\prod_{i = d+1}^{D}
\tilde{C}_i
\right) \norm{g_D(\yto{D})}_{\pi, 2^{D+1}}^{2^{D+1}} 
\left( 
\frac{2^{({2^{d+1}} - k_d(2^{d+1} - 1)})}{2^{({2^{d+1}} - k_d(2^{d+1} - 1))} - 1}
\right) \\
&=
\left( 
\prod_{i = d}^{D}
\tilde{C}_i
\right) 
\norm{g_D(\yto{D})}_{\pi, 2^{D+1}}^{2^{D+1}} .
\end{align*}

Here the choice $k_d \in \left( 1 , \frac{2^{d+1}}{2^{d+1} - 1} \right)$ is crucial. It ensures 
$2^{(2^{d+1})} (1-r_d)^{(2^{d+1} - 1)} > 1$, and in turn ensures the infinite summation of the above geometric series is finite.

\underline{Unbiasedness:}\\
Now we show the unbiasedness of $R_d(y^{(0 : d - 1)})$. Firstly, since we have just shown $R_d(\yto{d-1})$ has a finite $2^{d+1}$-th moment under $\pi$, it directly implies $\lvert R_d(\yto{d-1})\rvert$ has a finite first moment, which further implies $\bE[\lvert R_d(\yto{d-1})\rvert \mid \yto{d-1}]$ is finite for $\pi$-almost surely $\yto{d-1}$.

 We fix $\yto{d-1}$ from now on, and we will write $\bE_{\pi_{d:D}}[\cdot]$ as a shorthand notation for $\bE[ \cdot \mid \yto{d-1}]$
 . Recall that we have output $R_d = \Delta_{N_d} / p_{r_d}(N_d)$, with
\[
\Delta_{N_{d}}
=
g_{d}\left( 
\yto{d}, \frac{S_{2^{N_{d}}}}{2^{N_{d}}}
\right)
-
\frac12
\left[ 
g_{d}\left( 
\yto{d}, \frac{\So_{2^{{N_{d}} - 1}}}{2^{{N_{d}} - 1}}
\right)
+
g_{d}\left( 
\yto{d}, \frac{\Se_{2^{{N_{d}} - 1}}}{2^{{N_{d}} - 1}}
\right)
\right].
\]
Then, 
\begin{align*}
\bE_{\pi_{d : D}}
[R_d(y^{(0 : d - 1)})]
&=
\bE_{\pi_{d : D}}
\left[ 
\bE
\left[
\frac{\Delta_{N_d}}{p_{r_d}(N_d)}
\mid 
N_d
\right]
\right] \\
&=
\bE_{\pi_{d : D}} 
\left[ 
\sum_{n=0}^\infty 
\frac{\Delta_n}{p_{r_d}(n)} p_{r_d}(n)
\right] \\
&=
\bE_{\pi_{d : D}}
\left[ 
\sum_{n=0}^\infty 
\Delta_n
\right] \\
&\stackrel{(\star\star\star)}{=}
\sum_{n=0}^\infty
\bE_{\pi_{d : D}}[\Delta_n] \\
&=
\sum_{n=0}^\infty
\bE_{\pi_{d : D}}\left\{
g_{d}\left( 
\yto{d}, \frac{S_{2^{n}}}{2^{n}}
\right)
-
\frac12
\left[ 
g_{d}\left( 
\yto{d}, \frac{\So_{2^{{n} - 1}}}{2^{{n} - 1}}
\right)
+
g_{d}\left( 
\yto{d}, \frac{\Se_{2^{{n} - 1}}}{2^{{n} - 1}}
\right)
\right]
\right\} \\
&=
\sum_{n=1}^\infty
\left\{
\bE_{\pi_{d : D}}\left[ 
g_{d}\left(
\yto{d}, \frac{S_{2^{n}}}{2^{n}}
\right)
\right]
-
\bE_{\pi_{d : D}}\left[ 
g_{d}\left( 
\yto{d}, \frac{S_{2^{n-1}}}{2^{n-1}}
\right)
\right]
\right\}
+
\bE_{\pi_{d : D}}[g_{d}(\yto{d}, g_{d+1}(1))] \\
&=
\bE_{\pi_{d : D}}\left[ 
g_d\left( 
\yto{d}, \lim_{n \to \infty} \frac{S_{2^n}}{2^n}
\right)
\right] \\
&=
\bE_{\pi_{d : D}}\left[ 
g_d\left( 
\yto{d}, \gamma_{d+1}(\yto{d})
\right)
\right] \\
&=
\gamma_{d}(\yto{d-1}).
\end{align*}

All the above calculations are straightforward except for $(\star\star\star)$, which swaps the order of expectation and summation. Therefore we complete this proof of unbiasedness by justifying the swap in $(\star\star\star)$. To justify the swap, it suffices to show $\sum_n \bE[|\Delta_n|] < \infty$. Notice that $R_d(\yto{d})$ can be equivalently written as $\sum_{n=1}^\infty \Delta_n \bI(N_d = n)/p_{r_d}(n)$ where $N_d$ independent with $\{\Delta_i\}$. Calculating $\bE_{\pi_{d:D}}[\lvert R_d(\yto{d}) \rvert]$ yields:

\begin{align*}
    \bE_{\pi_{d:D}}[\lvert R_d(\yto{d-1}) \rvert] &= \bE_{\pi_{d:D}}\left[\left\lvert \sum_{n=1}^\infty \frac{\Delta_n \bI(N_d = n)} {p_{r_d}(n)}\right\rvert\right] \\
    &  = \bE_{\pi_{d:D}}\left[\sum_{n=1}^\infty \left\lvert  \frac{\Delta_n \bI(N_d = n)} {p_{r_d}(n)}\right\rvert\right] \quad\quad\quad\quad \textcolor{blue}{\text{only one term in the summation is non-zero}}\\ 
    & = \sum_{n=1}^\infty \bE_{\pi_{d:D}}\left[\left\lvert  \frac{\Delta_n \bI(N_d = n)} {p_{r_d}(n)}\right\rvert\right] \quad\quad\quad\quad \textcolor{blue}{\text{every term is non-negative}}\\
    & = \sum_{n=1}^\infty \bE_{\pi_{d:D}}\left[\left\lvert  \frac{\Delta_n} {p_{r_d}(n)}\right\rvert\right]\bE\left[\bI(N_d = n)\right] \quad \textcolor{blue}{\text{independence between $N$ and $\{\Delta_i\}$}}\\
    & = \sum_{n=1}^\infty \bE_{\pi_{d:D}}\left[\left\lvert \Delta_n\right\rvert \right].
\end{align*}
Since we already know $\bE_{\pi_{d:D}}[\lvert R_d(\yto{d-1})\rvert]<\infty$, this justifies our swap $(\star\star\star)$.

\end{proof}

\newpage
\section{Proof of Theorem \ref{thm:general D, Lipschitz case}}

Recall that  we say $\{g_d\}_{d=0}^{D-1}$ satisfy the last-component bounded Lipschitz condition (LBL) if for all $x, z \in \bR$:
\[
|g_d(\yto{d}, x) - g_d(\yto{d}, z)| < L_d|x - z|.
\]

The proof strategy of Theorem \ref{thm:general D, Lipschitz case} is very similar to Theorem \ref{thm:general D, second order case}. We start with a backward induction. 
\begin{proof}
\textbf{Case 1:~~}$d = D$\\

When $d = D$, Algorithm \ref{alg:recursive-rMLMC} samples one $y^{(D)} \sim \pi_{D}$ and outputs $R_D(\yto{D-1}) := g_D(\yto{D})$. Again, we first prove our output $R_D(\yto{D-1})$ has a finite expectation for almost  every fixed $\yto{D-1}$, then its expectation equals $\gamma_D(\yto{D-1})$ follows directly from the algorithm design. To show the first point, notice that the  expectation of $|R_D(\yto{D-1})|$  equals the conditional expectation $\bE[\lvert g_D(\yto{D})\rvert\mid \yto{D-1}]$. Since $\bE[|g_D|]<\infty$ by assumption, we have $$\bE[\lvert g_D(\yto{D})\rvert\mid \yto{D-1}] <\infty$$ almost surely. Therefore Algorithm \ref{alg:recursive-rMLMC} is unbiased when $d = D$ for almost  every input $\yto{D-1}$. Furthermore, it has computational cost $1$, and the output has finite $\left(2 - \frac{\delta}{2^{D}}\right)$-th moment.


\textbf{Case 2:~~}$0\leq d \leq D-1$

Now that our base case is proven, we proceed via backwards induction. Let $\delta_d := \delta/2^d$ for every $d \in \{0,1,\ldots, D\}$. Suppose unbiasedness, finite $(2 - \delta_{d+1})$-th moment, and finite expected computational cost are all satisfied for $d+1$ where $0 \leq d \leq D - 1$. Then  Algorithm \ref{alg:recursive-rMLMC} will call itself independently for $2^{N_{d}}$ times, each with input \{Depth index: $d+1$, Trajectory History: $H = \yto{d}$, Parameters: $r_{d+1}, \cdots, r_{D-1}$\}. This gives us i.i.d. samples $R_{d+1}(\yto{d})(1),...,R_{d+1}(\yto{d})(2^{N_{d}})$ which are used to compute the following:
\begin{align*}
S_{2^{N_{d}}}
&=
R_{d+1}(\yto{d})(1) +
R_{d+1}(\yto{d})(2) +
\cdots +
R_{d+1}(\yto{d})(2^{N_{d}}), \\
\So_{2^{N_{d} - 1}}
&=
R_{d+1}(\yto{d})(1) +
R_{d+1}(\yto{d})(3) +
\cdots +
R_{d+1}(\yto{d})(2^{N_{d}} - 1), \\
\Se_{2^{N_{d} - 1}}
&=
R_{d+1}(\yto{d})(2) +
R_{d+1}(\yto{d})(4) +
\cdots +
R_{d+1}(\yto{d})(2^{N_{d}}).
\end{align*}

Then Algorithm \ref{alg:recursive-rMLMC} returns as output $R_{d}(\yto{d}) = \Delta_{N_{d}} / p_{r_d}(N_{d})$, where $\Delta_{N_{d}}$ is defined in Algorithm \ref{alg:recursive-rMLMC}. By the inductive hypothesis on $d+1$, we have for almost  every $\yto{d}$:
\[
\bE
[
R_{d + 1}(\yto{d}) \mid \yto{d}
]
=
\gamma_{d + 1}(\yto{d})\]
and
\[
\bE_{\pi}\left[ |R_{d+1}( \yto{d})|^{2 - \delta_{d+1}} \right] 
< 
\left(
\prod_{i = d}^{D} \tilde{C}_i
\right)
\norm{g_D(\yto{D})}_{\pi, 2}^{2 - \delta_{d+1}}.
\]

We will start with showing  $R_d(y^{(0 : d - 1)})$ has a finite computational cost and a finite $(2 - \delta_d)$-th moment, and then show the unbiasedness. \\
\underline{Finite cost:}\\
To show the computational cost, recall that implementing Algorithm \ref{alg:recursive-rMLMC} with input depth $d$ requires $2^{N_d}$ calls of Algorithm \ref{alg:recursive-rMLMC} with input depth $d+1$. It suffices to check $r_d > 0.5$, which reduces to check the upper bound for $k_d$ (defined in Theorem \ref{thm:general D, Lipschitz case}) satisfies
\[\left( \frac{2^{d+2} - 3\delta}{2^{d+3} - 3\delta}\right) 
\left( \frac{2^{d+1} - \delta}{2^d - \delta}\right) > 1.
 \] 
Let $t \coloneqq 2^d$ and the above product becomes:
\[
\frac{4t - 3\delta}{8t - 3\delta} \frac{2t - \delta}{t-\delta} = \frac{8t^2 + 3\delta^2 - 10\delta}{8t^2 + 3\delta^2 - 11\delta} > 1.
\]

Since $N_d\sim \Geo(r_d)$ with $r_d > 0.5$, calling Algorithm \ref{alg:recursive-rMLMC} with input depth $d$ has an expected cost:
\begin{align*}
    \frac{r_d}{2r_d - 1} \times  \text{the expected cost of Algorithm \ref{alg:recursive-rMLMC} with input depth}~ d+1,
\end{align*}
where $\frac{r_d}{2r_d - 1} = \bE[2^{N_d}] < \infty$. By our inductive hypothesis, the second term in the above product is finite, therefore the expected cost of Algorithm \ref{alg:recursive-rMLMC} with input depth $d$ is also finite. 

\underline{Finite $(2 - \delta_d)$-th moment:}\\
Next we show $R_{d}$ has a finite $(2 - \delta_d)$-th moment. By the uniform $L_d$-Lipschitz property of $g_d$:
\begin{align*}
|\Delta_n|
&\leq 
\frac12 
\left| 
g_d\left( 
\yto{d}, \frac{S_{2^{n}}}{2^{n}}
\right) 
-
g_d\left( 
\yto{d}, \frac{\So_{2^{n - 1}}}{2^{n - 1}}
\right)
\right|
+
\frac12
\left| 
g_d\left( 
\yto{d}, \frac{S_{2^n}}{2^n}
\right) 
-
g_d\left( 
\yto{d}, \frac{\Se_{2^{n - 1}}}{2^{n - 1}}
\right)
\right|\\
&\leq 
\frac{L_{d}}{2}
\left| 
\frac{\So_{2^{n - 1}}}{2^{n - 1}}
-
\frac{\Se_{2^{n - 1}}}{2^{n - 1}}
\right|.
\end{align*}

Therefore, for any fixed $1 \leq p < 2$, applying the triangle inequality under the norm $\lVert \cdot \rVert_{\pi, p}$, and applying Lemma \ref{lem:conditional-MZ} with $Z_1 = \yto{d}$, $Z_2 = R_{d+1}(\yto{d})$, we have:

\begin{align*}
\norm{\Delta_n}_{\pi, p}
&\leq
\frac{L_{d}}{2}
\norm{
\frac{\So_{2^{n - 1}}}{2^{n - 1}}
-
\frac{\Se_{2^{n - 1}}}{2^{n - 1}}
}_{\pi, p} \\
&\leq
\frac{L_{d}}{2}
\norm{
\frac{\So_{2^{n - 1}}}{2^{n - 1}}
-
\gamma_{d+1}
}_{\pi, p}
+
\frac{L_{d}}{2}
\norm{
\gamma_{d+1}
-
\frac{\Se_{2^{n - 1}}}{2^{n - 1}}
}_{\pi, p} \\
&\leq
L_{d}
\left(
\frac{B_p \bE_{\pi}[|R_{d+1}(\yto{d})|^{p}]}{2^{(n-1)(p - 1)}}
\right)^{1/p},
\end{align*}
exponentiating both sides by  $p$ yields,
\[
\bE_{\pi}[|\Delta_n|^{p}]
\leq
\frac{L_d^p B_p\bE_{\pi}[|R_{d+1}(\yto{d})|^{p}]}{2^{(n-1)(p - 1)}}
\leq
\frac{C(d,p)\bE_{\pi}[|R_{d+1}(\yto{d})|^{p}]}{2^{(p - 1)n}},
\]
where $C(d,p) = L_d^p B_p 2^{1-p}$.

Recall that $\delta_d = \delta / 2^d$, let us choose $q_d = 2 - (\delta_d + \delta_{d+1})/2$. Since 
\[
\left( 1 , \left( \frac{q_d - 1}{q_d} \right) \left(\frac{2 - \delta_d}{1 - \delta_d}\right)\right)
= 
\left( 
1 ,
\left( 
\frac{2^{d+2} - 3\delta}{2^{d+3} - 3\delta}
\right) 
\left( 
\frac{2^{d+1} - \delta}{2^d - \delta}
\right)
\right),
\]
by definition of $k_d$ in the Theorem statement we have 
\[
k_d < \left( \frac{q_d - 1}{q_d} \right) \left(\frac{2 - \delta_d}{1 - \delta_d}\right).
\]

Now we estimate the $(2-\delta_d)$-th moment of $R_d$. An important trick in the calculation below is that we are not going to use the above estimate of $\bE_\pi[\lvert\Delta_n \rvert^p]$ directly on $p = 2-\delta_d$. Instead, we will first use H\"older's inequality, and then bound $\bE_\pi[|\Delta_n|^{q_d}]^{(2-\delta_d)/q_d}$ via the above estimate. It turns out the first way gives us an  order of $2^{-n(1-\delta_d)}$, while the latter is of order $2^{-n(q_d-1)(2-\delta_d)/q_d}$. Since the function $(x-1)(2-\delta_d)/x$ is increasing with $x$ when $x>1$, and equals $1-\delta_d$ when $x = 2-\delta_d$, we gain an extra factor $2^{-\Omega(1)n}$ by choosing $q_d > 2-\delta_d$ and use H\"older's inequality, which is important for establishing our main result. The detailed calculation is below:

\begin{align*}
\bE_{\pi}&[|R_d(\yto{d-1})|^{2 - \delta_d}] 
\leq
\sum_{n=0}^\infty
\frac{\bE_{\pi}[|\Delta_n|^{2 - \delta_d}]}{p_{r_d}(n)^{1 - \delta_d}} \\ 
&\leq
\sum_{n=0}^\infty
\frac{
\left(\bE_{\pi}\left[ |\Delta_n|^{2 - \delta_d \cdot \frac{q_d}{2-\delta_d}} \right]\right)^{\frac{2-\delta_d}{q_d}}
}{p_{r_d}(n)^{1 - \delta_d}} \quad\quad\quad\quad\quad\quad\quad\quad\quad\quad\quad\quad\quad\quad
\quad\quad\textcolor{blue}{\text{H\"older's inequality}}\\
&\leq
\frac{1}{r_d^{1 - \delta_d}}
\sum_{n=0}^\infty
\left(
\frac{C(d,q_d) \bE_{\pi}[|R_{d+1}(\yto{d})|^{q_d}]}{2^{(q_d - 1)n}}
\right)^{\frac{2 - \delta_d}{q_d}}
\frac{1}{(1 - r_d)^{(1 - \delta_d)n}} ~~~\quad~~~
\textcolor{blue}{\text{estimate of $\bE_\pi[|\Delta_n|^p]$ with $p = q_d$}} \\
&= 
C'(d) \norm{R_{d+1}(\yto{d})}^{2 - \delta_d}_{\pi, q_d}\sum_{n=0}^\infty 
\left(
\frac{1}{2^{\frac{(q_d - 1)}{q_d}(2 - \delta_d) - k_d(1 - \delta_d)}}
\right)^n \quad\qquad~~\quad~~~ \textcolor{blue}{\text{here}~C'(d) = \frac{C(d,q_d)^{(2-\delta_d)/q_d}}{r_d^{1 - \delta_d}}}
\\
&\leq
C'(d) \norm{R_{d+1}(\yto{d})}^{2 - \delta_{d}}_{\pi, 2 - \delta_{d+1}}
\sum_{n=0}^\infty 
\left(
\frac{1}{2^{\frac{(q_d - 1)}{q_d}(2 - \delta_d) - k_d(1 - \delta_d)}}
\right)^n  \quad\quad\quad~~~ \textcolor{blue}{\text{since}~q_d < 2-\delta_{d+1}}\\
&\leq
C'(d)
\left(
\prod_{i = d + 1}^{D} \tilde{C}_i
\right)
\norm{g_D(\yto{D})}_{\pi, 2}^{2 - \delta_d}
\left(
\frac{2^{\frac{(q_d - 1)}{q_d}(2 - \delta_d) - k_d(1 - \delta_d)}}{2^{\frac{(q_d - 1)}{q_d}(2 - \delta_d) - k_d(1 - \delta_d)} - 1}
\right) \quad \textcolor{blue}{\text{inductive hypothesis}} \\
&=
\left(
\prod_{i = d}^{D} \tilde{C}_i
\right)
\norm{g_D(\yto{D})}_{\pi, 2}^{2 - \delta_d},
\end{align*}
and note the RHS is still finite given the assumption of our theorem on $g_D$.
Again, as we can see in the proof, the choice of $k_d$ and $q_d$ is crucial for our calculation. It ensures $\frac{(q_d - 1)}{q_d}(2 - \delta_d) - k_d(1 - \delta_d) > 0$, and in turn ensures the above summation of the geometric series  converges.

\underline{Unbiasedness:}\\

The proof of unbiasedness of our estimator in this case is identical to the LBS case, however we still require a justification of the existence of a finite conditional expectation of $R_d(\yto{d-1})$. By what we have just proven, 
\[
\bE_{\pi}[|R_d(\yto{d-1})|]
\leq 
\left(
\bE_{\pi}
\left[
|R_d(\yto{d-1})|^{2 - \frac{\delta}{2^d}}
\right]
\right)^{1/(2 - \frac{\delta}{2^d})}
<
\infty.
\]
Given $\bE_{\pi}[|R_d(\yto{d-1})|] < \infty$, we immediately have $\bE_{\pi_{d : D}}
[R_d(y^{(0 : d - 1)})]$ exists for almost  every $\yto{d-1}$.
 
\end{proof}

\section{Construction of the NMC estimator}\label{sec: NMC}
The construction of the NMC estimator is described in \cite{rainforth2018nesting}. For concreteness, we explain the construction details here for the $D= 2$ case, which we use in Section \ref{sec:numerical}. 

Fix positive integers $N_0, N_1, N_2$, users first simulate $N_0$ i.i.d. $\{\y0_i\}_{i=1}^{N_0} \sim \pi_0$. For each fixed $\y0_i$, users sample $N_1$ i.i.d. $\{\y1_{i,j}\}_{j=1}^{N_1}$ from $\pi_1(\cdot \mid \y0_i)$. Then for each fixed trajectory $(\y0_i, \y1_{i,j})$, users sample $N_2$ i.i.d. $\{\y2_{i,j,k}\}_{k=1}^{N_2}$ from $\pi_2(\cdot \mid \y0_i, \y1_{i,j})$. After getting all these samples, we use the standard estimator:
\[
\hat{\gamma}_2(\y0_i,\y1_{i,j}) \coloneqq \frac{1}{N_2} \sum_{k=1}^{N_2}g_2\left(\y0_i,\y1_{i,j}, \y2_{i,j,k}\right)
\]
to estimate $\gamma_2(\y0_i,\y1_{i,j})$. Then using the plug-in estimator
\[
\hat{\gamma}_1(\y0_i) \coloneqq \frac{1}{N_1} \sum_{j=1}^{N_1}g_1\left(\y0_i,\y1_{i,j}, \hat{\gamma}_2(\y0_i,\y1_{i,j})\right)
\]
to estimate $\gamma_1(\y0_i)$. Finally, plugging in these $N_0$ estimators $\{\hat{\gamma}_1(\y0_i)\}_{i=1}^{N_0}$ to form the NMC estimator
\[
\hat{\gamma}_0 \coloneqq \frac{1}{N_0} \sum_{i=1}^{N_0}g_0\left(\y0_i,
\hat{\gamma}_1(\y0_i)\right)
\]
for $\gamma_0$.

It is proven in \cite{rainforth2018nesting} that when all $N_0, N_1, N_2$ go to $\infty$, $\hat\gamma_0$ converges to $\gamma_0$. It remains crucial to allocate $N_0,N_1,N_2$ to maximize the convergence rate with respect to the total sample size $n = N_0 N_1 N_2$. The choice $N_0 = N_1^2 = N_2^2$ is suggested in \cite{rainforth2018nesting}, which has a $\calO(N^{-1/4})$ convergence rate for the rMSE, or a $\calO(N^{-1/2})$  rate for the MSE.

\section{Additional statistics of Section \ref{sec:numerical}}\label{sec: additional statistics}
Figure \ref{fig:error_comparison} in Section \ref{sec:numerical} compares the errors between $\READ$, NMC1, and NMC2 in terms of the cost of the total sample size. Here we also compare their estimation errors in terms of the wall-clock time. For $\READ$, we call Algorithm \ref{alg:recursive-rMLMC} repeatedly for $10^5$ times. For NMC1, we choose $N_0 = N_1 = N_2 = 400$. For NMC2, we choose $N_0 = N_1^2 = N_2^2 = 10^4$. Their estimation errors and the corresponding wall-clock time costs are summarized in Table \ref{tab:clock-time cost}. It is  clear that $\READ$ is both faster (in wall-clock time) and more accurate. We also calculate the time-normalized squared error, defined as the product between the time cost and the squared error \cite{glynn1992asymptotic}. From the normalized squared error, $\READ$ is more than $130$ times more efficient than NMC2, and more than $407$ times more efficient than NMC1. 

\begin{table}[htbp]
\begin{center}
    
\begin{tabular}{|c|ccccc|}
\hline
Method & Setting                      & Total Sample Cost     & Time/Seconds & Squared Error         & \begin{tabular}[c]{@{}c@{}}Time-normalized \\ Squared Error\end{tabular} \\ \hline
$\READ$ & $10^5$ repetitions           & $4.625 \times 10^{5}$ & $\color{red}{11.86}$      & $\color{red}{3.186\times 10^{-6}}$ & $\color{red}{3.78 \times 10^{-5}}$                                                    \\
NMC1   & $N_0 = N_1 = N_2 = 400$      & $6.4 \times 10^7$     & $43.79$      & $3.51 \times 10^{-4}$ & $1.54 \times 10^{-2}$                                                    \\
NMC2   & $N_0 = N_1^2 = N_2^2 = 10^4$ & $10^8$                & $72.6$       & $6.8 \times 10^{-5}$  & $4.93 \times 10^{-3}$                                                    \\ \hline
\end{tabular}\caption{Cost comparison between different methods}\label{tab:clock-time cost}
\end{center}

\end{table}

\section{Extra experiments}\label{sec:extra experiment}

We consider an extra experiment with unknown ground truth. Let $\sigma(x) \coloneqq e^x/(1+e^x)$ be the sigmoid function. 
Suppose the process $(\y0,\y1,\y2)$ satisfies $\y0\sim \bN(0,1), \y1 \sim \bN(\y0,1), \y2\sim \bN(\y1,1)$. Define $g_0(\y0, z) \coloneqq \sigma(\y0+z), g_1(\yto1, z) \coloneqq \sigma(\y1 + z)$, and $g_2(\yto2) \coloneqq \sigma(\y2)$. The target quantity $\gamma_0$ defined \eqref{eqn:nested target} is again a nested expectation with  $D = 2$. Although we can not analytically calculate out $\gamma_0$,   we still implement our $\READ$ estimator with the NMC estimators in \cite{rainforth2018nesting} and compare their performance. The parameters of $\READ$ are the same as Section \ref{sec:numerical}, the allocation of $N_0,N_1,N_2$ of the NMC estimators also follows the same way as Section \ref{sec:numerical}.

The scatter plot of the estimation results is shown in Figure \ref{fig:sigmoid_estimation}. Although no ground truth is available, the trend for the estimation is clear. All three methods eventually get close to $0.612$, represented by the dotted black line in Figure \ref{fig:sigmoid_estimation}. It is also clear from the plot that $\READ$ always stays very close to the black line. In contrast, both NMC1 and NMC2 are significantly more fluctuated than $\READ$, where NMC1 appears to be the most unstable estimator. This again matches with the theoretical predictions in our paper and \cite{rainforth18nesting_prob} that $\READ$ converges the fastest while NMC1 converges the slowest.
\begin{figure}[htbp]
\begin{center}
\centerline{\includegraphics[width=\columnwidth]{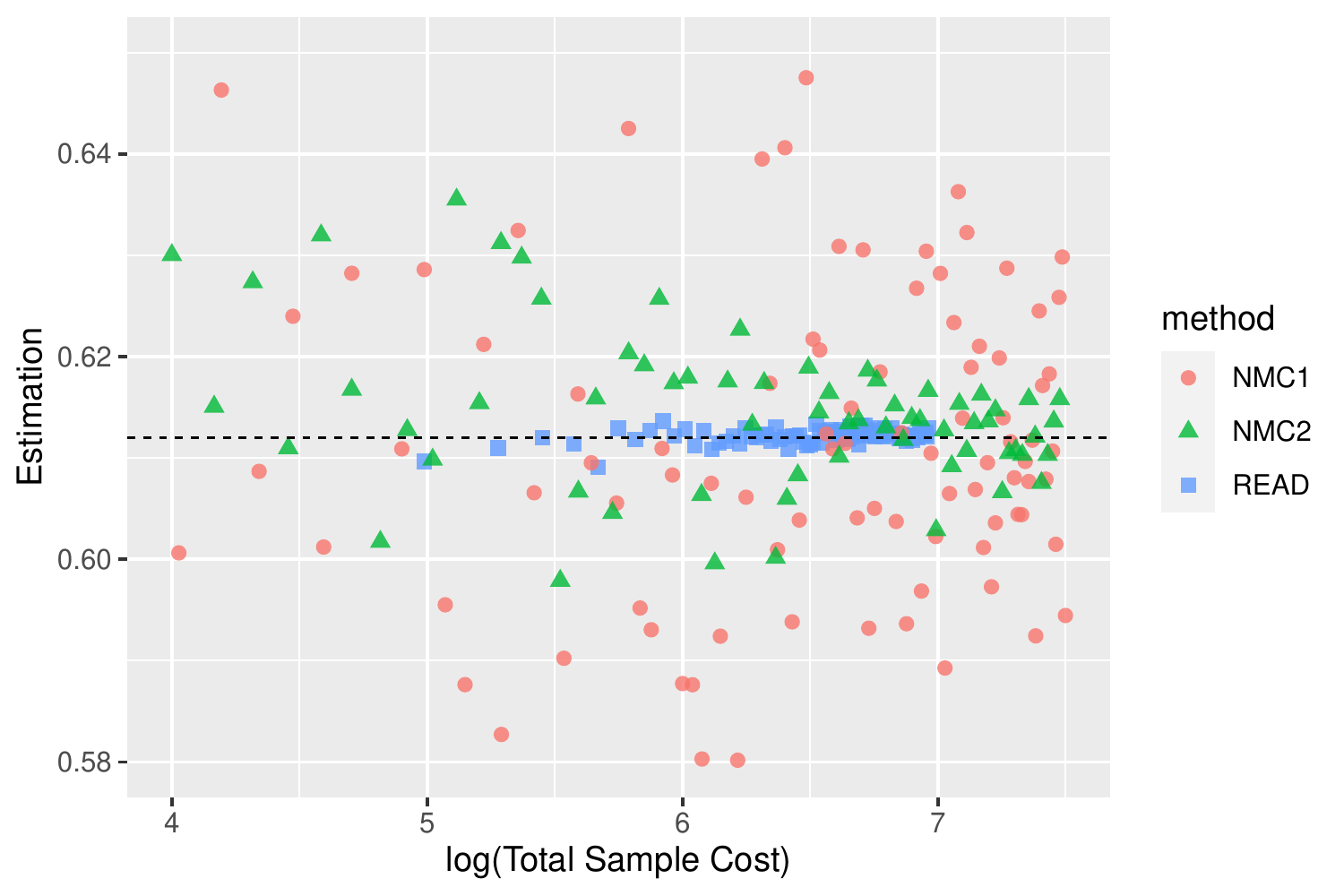}}
\caption{Scatterplot of the estimation of $\gamma_0$ as a function of $\log_{10}(\text{Total Sample Cost})$. Blue, red, green points correspond to $\READ$, NMC$1$, NMC$2$ estimators respectively.}
\label{fig:sigmoid_estimation}
\end{center}
\end{figure}
\begin{figure}
    \centering
    \subfigure[]{\includegraphics[width=0.45\textwidth]{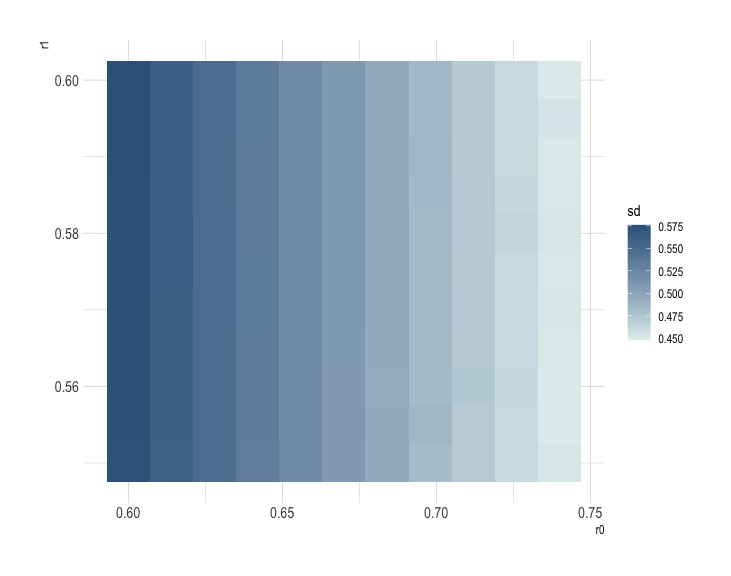}} 
    \subfigure[]{\includegraphics[width=0.45\textwidth]{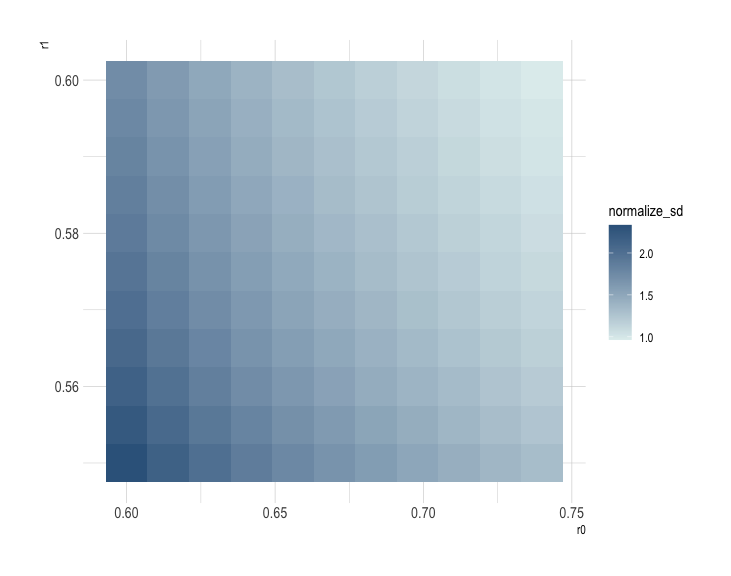}} 
    \caption{(a): Heatmap of the (empirical) standard deviation of $\READ$. (b): Heatmap of the work-normalized standard deviation of $\READ$. Here $r_0\in (0.6, 0.74)$, $r_1\in (0.55,0.6)$. Each standard deviation is estimated based on $10^6$ repetitions of Algorithm \ref{alg:recursive-rMLMC}.}
    \label{fig:heatmap_small_r}
\end{figure}

Next we let the parameters $(r_0,r_1)$ in Algorithm \ref{alg:recursive-rMLMC} vary and investigate the proper choice of the parameters in this experiment. Theorem \ref{thm:general D, second order case} shows any $(r_0,r_1) \in (0.5,0.75) \times (0.5, 1 - 2^{-4/3})$ guarantees $\READ$ has finite variance and finite cost. Therefore we choose $r_0$ on the lattice $\{0.6, 0.614,\ldots,  0.74\}$ and $r_1$ on the lattice $\{0.55, 0.555, \ldots, 0.6\}$. For each pair of $(r_0, r_1)$, we repeat Algorithm \ref{alg:recursive-rMLMC} for $10^6$ times, record the results and calculate their empirical standard deviation. The heatmap is shown in the left plot of Figure \ref{fig:heatmap_small_r}. The pattern suggests the standard deviation depends crucially on the choice of $r_0$, but less on $r_1$. The standard deviation decreases when $r_0$ increases. 

Since the expected sample cost of Algorithm \ref{alg:recursive-rMLMC}
equals $\left(r_1/(2r_1 - 1)\right) \left(r_2/(2r_2 - 1)\right)$. We also plot the `work-normalized standard deviation', which is defined as $\sqrt{\text{Expected Sample Cost}} \times \text{Standard deviation}$ in \cite{glynn1992asymptotic} to measure the efficiency of difference choices of $(r_0,r_1)$. The heatmap is shown in the right subplot of Figure \ref{fig:heatmap_small_r}. Our result suggests users should choose larger values of $(r_0,r_1)$  to maximize the efficiency, at least in this example.

Finally we test our results when $r_0, r_1$ are both beyond the range given by Theorem \ref{thm:general D, second order case}. Algorithm \ref{alg:recursive-rMLMC} can still be implemented, though there is no guarantees on the finite variance. Nevertheless, we choose $r_0 \in \{0.8,0.81,\ldots, 0.9\}$ and $r_1 \in \{0.7,0.71,\ldots, 0.8\}$ and report the heatmaps of the standard deviations/work-normalized standard deviations in Figure s\ref{fig:heatmap_large_r}. The estimates become significantly less stable, as some pairs of $(r_0,r_1)$ have much larger standard deviation than their neighborhoods. This suggests the actual standard deviation maybe already infinity (though the empirical standard deviation will always be finite), and therefore our result is less reliable. In conclusion, although larger values of $(r_0,r_1)$ can reduce the average cost of each implementation, users should not choose them too large as it may sacrifice the finite variance. Users can choose the parameters closer to the upper end of the ranges in Theorem \ref{thm:general D, second order case}, but not exceed these ranges.
\begin{figure}[htbp]
    \centering
    \subfigure[]{\includegraphics[width=0.45\textwidth]{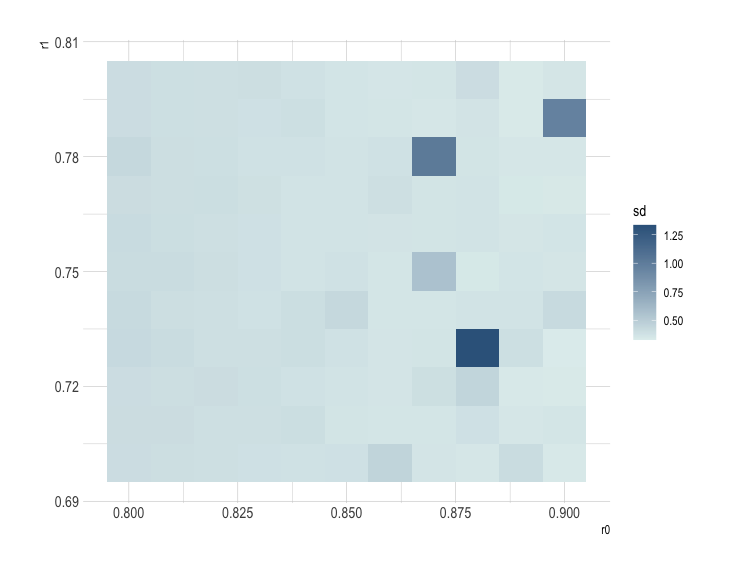}} 
    \subfigure[]{\includegraphics[width=0.45\textwidth]{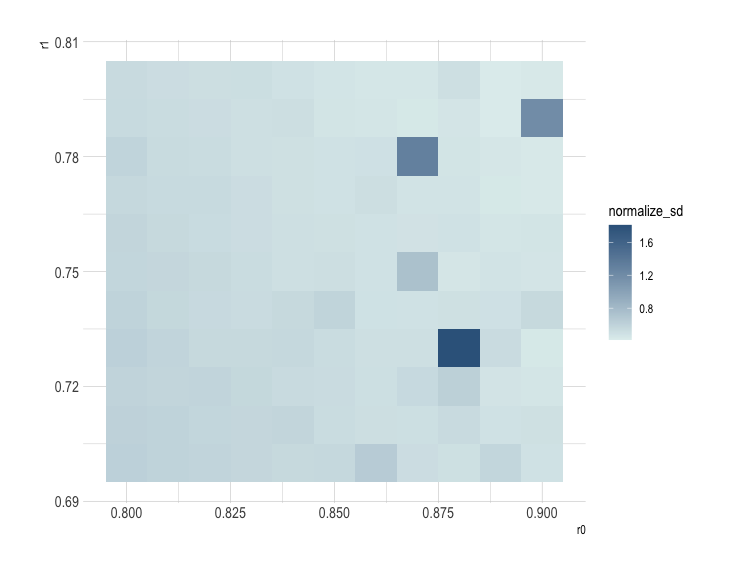}} 
    \caption{(a): Heatmap of the (empirical) standard deviation of $\READ$. (b): Heatmap of the work-normalized standard deviation of $\READ$. Here $r_0\in (0.8, 0.9)$, $r_1\in (0.7,0.8)$. Each standard deviation is estimated based on $10^6$ repetitions of Algorithm \ref{alg:recursive-rMLMC}.}
    \label{fig:heatmap_large_r}
\end{figure}

\end{document}